\documentclass[11pt]{article}
\newcommand{\TwoOrOneColumn}{one}
\usepackage[letterpaper,margin=1.00in]{geometry}
\usepackage{amsmath, amssymb, amsthm, amsfonts}
\usepackage{bbm}

\usepackage{ifthen}
\usepackage{tikz}
\usetikzlibrary{positioning,decorations.pathreplacing}

\usepackage{cite}
\usepackage{appendix}
\usepackage{graphicx}
\usepackage{color}
\usepackage{color}
\usepackage{algorithm}
\usepackage{algorithm}
\usepackage[noend]{algpseudocode}
\usepackage{epstopdf}
\usepackage{wrapfig}
\usepackage{paralist}
\usepackage{wasysym}
\usepackage[textsize=tiny]{todonotes}

\usepackage{framed}
\usepackage[framemethod=tikz]{mdframed}
\usepackage[bottom]{footmisc}
\usepackage{enumitem}
\setitemize{noitemsep,topsep=3pt,parsep=3pt,partopsep=3pt}
\usepackage[font=small]{caption}
\usepackage{xspace}

\newtheorem{theorem}{Theorem}[section]
\newtheorem{lemma}[theorem]{Lemma}
\newtheorem{meta-theorem}[theorem]{Meta-Theorem}

\newtheorem{remark}[theorem]{Remark}
\newtheorem{corollary}[theorem]{Corollary}
\newtheorem{proposition}[theorem]{Proposition}
\newtheorem{observation}[theorem]{Observation}
\newtheorem{definition}[theorem]{Definition}

\definecolor{darkgreen}{rgb}{0,0.5,0}
\usepackage{hyperref}
\hypersetup{
    unicode=false,          % non-Latin characters in Acrobat’s bookmarks
    colorlinks=true,        % false: boxed links; true: colored links
    linkcolor=red,          % color of internal links (change box color with linkbordercolor)
    citecolor=darkgreen,        % color of links to bibliography
    filecolor=magenta,      % color of file links
    urlcolor=cyan           % color of external links
}
\usepackage[capitalize, nameinlink]{cleveref}
\Crefname{remark}{Remark}{Remarks}
\Crefname{observation}{Observation}{Observations}
\algnewcommand\algorithmicswitch{\textbf{switch}}
\algnewcommand\algorithmiccase{\textbf{case}}

\algdef{SE}[SWITCH]{Switch}{EndSwitch}[1]{\algorithmicswitch\ #1\ \algorithmicdo}{\algorithmicend\ \algorithmicswitch}%
\algdef{SE}[CASE]{Case}{EndCase}[1]{\algorithmiccase\ #1}{\algorithmicend\ \algorithmiccase}%
\algtext*{EndSwitch}%
\algtext*{EndCase}%
\newcommand{\eps}{\varepsilon}

\newcommand{\congest}{$\mathsf{CONGEST}$\xspace}
\newcommand{\local}{$\mathsf{LOCAL}$\xspace}

\newcommand{\poly}{\operatorname{\text{{\rm poly}}}}

\newcommand{\aglp}{$2^{O\left(\sqrt{\log n \log \log n}\right)}$}
\newcommand{\ps}{$2^{O\left(\sqrt{\log n}\right)}$}

\renewcommand{\paragraph}[1]{\vspace{0.15cm}\noindent {\bf #1}:}

\newcommand\blfootnote[1]{%
  \begingroup
  \renewcommand\thefootnote{}\footnote{#1}%
  \addtocounter{footnote}{-1}%
  \endgroup
}

\newcommand{\FullOrShort}{full}

\ifthenelse{\equal{\FullOrShort}{full}}{
	
  \newcommand{\fullOnly}[1]{#1}
  \newcommand{\shortOnly}[1]{}
}{
    \newcommand{\shortOnly}[1]{#1}
    \newcommand{\fullOnly}[1]{}
  }

\begin{document}
\date{}
\title{Polylogarithmic-Time Deterministic Network Decomposition \\ and Distributed Derandomization }
\author{V\'{a}clav Rozho\v{n} \\
\small ETH Zurich \\
\small rozhonv@student.ethz.ch
\and
Mohsen Ghaffari\\
\small ETH Zurich \\
\small ghaffari@inf.ethz.ch
}
\maketitle

\begin{abstract}
\fullOnly{\blfootnote{This project has received funding from the European Research Council (ERC) under the European Union's Horizon 2020 research and innovation programme (grant agreement No. 853109).}}
We present a simple polylogarithmic-time deterministic distributed algorithm for network decomposition. This improves on a celebrated $2^{O(\sqrt{\log n})}$-time algorithm of Panconesi and Srinivasan [STOC'92] and settles a central and long-standing question in distributed graph algorithms. It also leads to the first polylogarithmic-time deterministic distributed algorithms for numerous other problems, hence resolving several well-known and decades-old open problems, including Linial's question about the deterministic complexity of maximal independent set [FOCS'87; SICOMP'92]---which had been called the most outstanding problem in the area. 

\bigskip
The main implication is a more general distributed derandomization theorem: Put together with the results of Ghaffari, Kuhn, and Maus [STOC'17] and Ghaffari, Harris, and Kuhn [FOCS'18], our network decomposition implies that 
\begin{center}
$\mathsf{P}$-$\mathsf{RLOCAL}$ = $\mathsf{P}$-$\mathsf{LOCAL}$.  \end{center}

 \noindent That is, for any problem whose solution can be checked deterministically in polylogarithmic-time, any polylogarithmic-time randomized algorithm  can be derandomized to a polylogarithmic-time deterministic algorithm. Informally, for the standard first-order interpretation of efficiency as polylogarithmic-time, distributed algorithms do not need randomness for efficiency.

\bigskip
By known connections, our result leads also to substantially faster \emph{randomized} distributed algorithms for a number of well-studied problems including $(\Delta+1)$-coloring, maximal independent set, and Lov\'{a}sz Local Lemma, as well as   \emph{massively parallel} algorithms for $(\Delta+1)$-coloring. 
\end{abstract}
\setcounter{page}{0}
\thispagestyle{empty}
% \listoftodos
\newpage

\section{Introduction}
We present a polylogarithmic-time deterministic distributed algorithm for network decomposition. This leads to substantially faster \emph{deterministic} and \emph{randomized} algorithms for many well-studied problems in distributed graph algorithms, as well as a general and efficient distributed derandomization theorem. These resolve several central open problems in the area.

\subsection{Background and State of the Art}
\paragraph{Model} We work with the standard model of distributed computing called \local~\cite{linial1987LOCAL, linial92}: The communication network is abstracted as an $n$-node graph $G=(V, E)$, with one processor on each node $v\in V$ which has a unique $\Theta(\log n)$-bit identifier. Communication happens in synchronous rounds, where per round each node can send one message, of potentially unbounded size, to each neighbor. In the \congest variant of the model\cite{peleg00}, each message can have $O(\log n)$ bits. 
At the beginning, each processor knows only its neighbors, and some estimates of global parameters, e.g., a polynomial upper bound on $n$. At the end, each processor should know its own part of the output, e.g., its color in the vertex coloring problem.

\medskip

\paragraph{State of the art} Prior to this work, the state of the art in distributed graph algorithms exhibited a significant (often nearly-exponential) gap between randomized and deterministic distributed algorithms. This gap constituted one of the foundational and long-standing questions in distributed algorithms. A well-known special case is an open question of Linial\cite{linial1987LOCAL, linial92} about the maximal independent set (MIS) problem: 
\begin{center}
\ifthenelse{\equal{\TwoOrOneColumn}{one}}{
    \begin{minipage}{0.8\textwidth}
}{
    \begin{minipage}{0.4\textwidth}
}

``\emph{can it} [MIS] \emph{always be found} [deterministically] \emph{in polylogarithmic time?}"
\end{minipage}
\end{center}
\medskip

This has been described as ``probably the most outstanding open problem in the area''\cite[Open Problem 11.2]{barenboimelkin_book}. Prior to our work, the best known deterministic algorithm had a round complexity of $2^{O(\sqrt{\log n})}$, by Panconesi and Srinivasan\cite{panconesi-srinivasan}. This should be contrasted with the beautiful $O(\log n)$-time randomized algorithms of Luby\cite{luby86} and Alon, Babai, and Itai\cite{alon86}. 

There is an abundance of similar open questions about obtaining polylogarithmic-time deterministic algorithms for other graph problems that admit polylogarithmic-time randomized algorithms; this includes $(\Delta+1)$-coloring, Lov\'{a}sz Local Lemma, defective colorings, hypergraph matching, sparse neighborhood covers, etc. Indeed, in the Conclusion and Open Problems chapter of their 2013 book, Barenboim and Elkin\cite[Chapter 11]{barenboimelkin_book} write:
\begin{center}

\ifthenelse{\equal{\TwoOrOneColumn}{one}}{
    \begin{minipage}{0.8\textwidth}
}{
    \begin{minipage}{0.4\textwidth}
}

\medskip
``\emph{Perhaps the most fundamental open problem in this field is to understand the power and limitations of randomization.}"
\medskip
\end{minipage}
\end{center}

\noindent They then continue to ask for a general derandomization technique:

\begin{center}
\ifthenelse{\equal{\TwoOrOneColumn}{one}}{
    \begin{minipage}{0.8\textwidth}
}{
    \begin{minipage}{0.4\textwidth}
}
\medskip
\emph{\textbf{Open Problem 11.1} Develop a general derandomization technique for the distributed message-passing model.}
\medskip
\end{minipage}
\end{center}

\noindent This generic open problem is followed by 16 concrete open problems, 7 of which ask for polylogarithmic-time (sometimes just called efficient) deterministic algorithms for various graphs problems that are known to admit efficient randomized algorithms. We note that a few of these concrete open problems were well-known, and they had been mentioned throughout the literature since the 1990s.

\subsection{Our Contribution}
In this paper, we answer all the concrete questions mentioned above by providing the first polylogarithmic-time deterministic algorithms for them. In fact, we show a more general distributed derandomization theorem, which proves the following:
\begin{theorem}[\textbf{LOCAL Derandomization Theorem}]\label{thm:derandomization} 
We have $$\mathsf{P}\textit{-}\mathsf{LOCAL} = \mathsf{P}\textit{-}\mathsf{RLOCAL}.$$ Here, $\mathsf{P}$-$\mathsf{LOCAL}$ denotes the family of {locally checkable problems}\footnote{To make our derandomization theorem stronger and more widely applicable, we use a relaxed version of local checkability: we call a problem \emph{locally checkable} if its solution can be checked deterministically in $poly(\log n)$ rounds, such that if the solution is incorrect, at least one node knows. Thus, each constraint of the problem spans a neighborhood of at most $poly(\log n)$ rounds. Notice that this readily includes problems such as MIS, coloring, etc. For a precise definition of locally checkable problems (but bounded to constant radius), we refer to \cite{naor95}.} that can be solved by deterministic algorithms in $\poly(\log n)$ rounds of the \local model in $n$-node graphs and $\mathsf{P}$-$\mathsf{RLOCAL}$ denotes the family of \emph{locally checkable problems} that can be solved by randomized algorithms in $\poly(\log n)$ rounds of the \local model, with success probability $1-1/n$.
\end{theorem}

Informally, if we follow the standard of viewing a $\poly(\log n)$-round algorithm as \emph{efficient}\footnote{This is similar to viewing a centralized algorithm with $\poly(n)$ time complexity or a parallel (PRAM model) algorithm with $\poly(\log n)$ time complexity as efficient.} (see e.g.\cite{barenboimelkin_book, linial92, panconesi-srinivasan}), \Cref{thm:derandomization} tells us that distributed algorithms in the \local model do not need randomness for efficiency. This holds for any \emph{locally checkable problem}, i.e., any problem for which the solution can be checked efficiently deterministically\footnote{We remark that the vast majority of the problems studied in the \local model throughout the literature are locally checkable. Moreover, such a restriction to locally checkable problems is necessary and the statement cannot hold for arbitrary problems, for trivial reasons: e.g., marking arbitrary $\Theta(\sqrt{n})$ nodes can be done in zero rounds by randomized algorithms but can be shown to require $\Omega(\sqrt{n})$ rounds for any deterministic algorithm\cite{ghaffari2018derandomizing}.}.
\medskip

At the heart of our derandomization result, and as the main novelty of this paper, we provide the first $\poly(\log n)$-round deterministic algorithm for network decomposition:

\begin{theorem} [\textbf{Network Decomposition Algorithm}]\label{thm:decomp-informal}
There is a deterministic distributed algorithm that in any $n$-node network $G=(V, E)$, in $\poly(\log n)$ rounds of the \local model, partitions the vertices into $O(\log n)$ disjoint color classes $V_1$, \dots, $V_{O(\log n)}$, such that in the subgraph $G[V_i]$ induced by the vertices of each color $i$, each connected component has diameter $O(\log n)$.
\end{theorem}
We prove \Cref{thm:decomp-informal} in \Cref{sec:algorithm}.
We note that prior to our work, the best known deterministic network decomposition had a round complexity of \ps, due to a celebrated work of Panconesi and Srinivasan\cite{panconesi-srinivasan}. 
This itself was an improvement on a \aglp-round distributed algorithm, presented by Awerbuch et al.\cite{awerbuch89}, in their pioneering work that defined network decomposition and showed its applications for distributed algorithms.

%In \Cref{sec:appendix_relation} \marginpar{!}we sketch the proof of their results, as well as an alternative new way of getting an algorithm of round complexity \aglp. 
%Then we compare these approaches with the approach we use to prove \Cref{thm:decomp-informal}. 

Our derandomization result, stated in \Cref{thm:derandomization}, follows by putting our new network decomposition, as stated in \Cref{thm:decomp-informal}, together with the derandomization framework developed by Ghaffari, Harris, and Kuhn~\cite{ghaffari2018derandomizing} and Ghaffari, Kuhn, and Maus~\cite{ghaffari2017complexity}. 

\paragraph{Implications} Through known connections, this derandomization leads to better \emph{deterministic} and \emph{randomized} distributed algorithms for a long list of well-studied problems. 
A sampling of the end-results includes (I) $\poly(\log n)$-round deterministic algorithms for maximal independent set, $\Delta+1$ coloring, the Lov\'{a}sz Local Lemma, and defective coloring, as well as (II) a $\poly(\log\log n)$-time randomized $\Delta+1$ coloring \cite{chang2018optimal}, a $\poly(\log\log n)$-time randomized algorithm for Lov\'{a}sz Local Lemma in constant degree graphs \cite{ghaffari2018derandomizing}, and an automatic complexity speed-up theorem from $o(\log n)$ to $\poly(\log\log n)$ in constant-degree graphs, for \emph{any} problem whose solution can be checked in $O(1)$ rounds\cite{chang2017time}. We discuss the implications in \Cref{sec:implications}. 

\subsection{An Overview of Our Network Decomposition Method}
Our network decomposition algorithm is surprisingly simple. Next, we briefly recall the previous methods\cite{panconesi-srinivasan, awerbuch89} and then give a quick outline of our construction:

\paragraph{A recap on the previous constructions} In the paper that introduced the concept of network decomposition, Awerbuch et al.~\cite{awerbuch89} provide a deterministic algorithm that computes a network decomposition with clusters of diameter $2^{O(\sqrt{\log n \log\log n})}$, which are colored with $2^{O(\sqrt{\log n \log\log n})}$ colors, in $2^{O(\sqrt{\log n \log\log n})}$ rounds. In a nutshell, their algorithm is based on a hierarchical clustering. We start with each node being its own cluster. Over time, iteratively, we merge clusters together, in a manner that each final cluster has $2^{O(\sqrt{\log n \log\log n})}$ neighboring clusters, and thus the clusters can be easily colored with $2^{O(\sqrt{\log n \log\log n})}$ colors. Per iteration, we locally \emph{group} clusters that have a ``high" degree --- more than $2^{O(\sqrt{\log n \log\log n})}$ neighboring clusters --- around some selected clusters (chosen using a ruling set algorithm). Then, in each group, we merge all the clusters into one cluster. The center clusters are chosen using a \emph{ruling set} procedure that ensures that the center clusters are somewhat far apart (concretely, at least 3 hops, in the \emph{cluster graph} that connects any two clusters that have adjacent nodes), while any high-degree cluster has a center within a small distance (concretely, $O(\log n)$ hops, in the cluster graph). Due the separation and the high degrees, each merge is formed by grouping together at least $2^{\Theta(\sqrt{\log n \log\log n})}$ clusters. Hence, we finish in $O(\sqrt{\log n/\log\log n})$ iterations. Per iteration, each cluster has diameter at most $O(\log n)$ times the diameter of the previous clusters, and thus within $O(\sqrt{\log n/\log\log n})$ iteration, each cluster diameter grows to be at most $2^{O(\sqrt{\log n \log\log n})}$. The algorithm of Panconesi and Srinivasan \cite{panconesi-srinivasan} follows the same outline but replaces the ruling set procedure with a maximal independent set procedure (of a constant power of the cluster graph), computed by a clever and careful recursive idea. This replaces the $O(\log n)$ growth factor in the diameter per iteration with $O(1)$. Then, re-optimizing the parameters to take advantage of this change improves the bounds to give a network decomposition with clusters of diameter $2^{O(\sqrt{\log n})}$, which are colored with $2^{O(\sqrt{\log n})}$ colors, in $2^{O(\sqrt{\log n})}$ rounds. 

\fullOnly{In \Cref{{sec:appendix_relation}}, we provide a new method for constructing a network decomposition, which also achieves such a round complexity of $2^{O(\sqrt{\log n \log\log n})}$. We then also discuss how both of these twp approaches, which appear to be fundamentally different, cannot go below the complexity of $2^{O(\sqrt{\log n})}$, even on a particular simple and well-structured graph.} 

\paragraph{Our construction, in a nutshell} The main part of our result is to obtain 
a network decomposition with clusters of diameter $poly(\log n)$, which are colored with $poly(\log n)$ colors, in $poly(\log n)$ rounds. We provide a surprisingly simple algorithm for this. We can later transform this construction to improve the first two parameters to $O(\log n)$. Similar to the previously outlined methods, our algorithm also forms the clusters iteratively. However, unlike the hierarchical clusterings of \cite{awerbuch89,panconesi-srinivasan}---where per iteration each new cluster is formed by merging a few of the nearby clusters of the previous iterations---during our construction, we \emph{release} some clusters and allow each of their individual nodes to make an independent decision on which adjacent cluster to join; some of these nodes can also remain in their initial cluster, or \emph{die}. Throughout the process, we ensure that at most a constant fraction of vertices die. Thus, via $O(\log n)$ repetition, each time by resurrecting the dead vertices and repeating the process on them, we can cluster all vertices. The decision of joining a neighboring cluster or dying is done in a manner that balances a few desirable properties, as we outline next.

The clustering process has $b$ phases, where $b=O(\log n)$ denotes the number of bits in the identifiers. We start with a trivial clustering where each (remaining) vertex is one cluster, on its own. Each cluster is identified with the node identifier of its center vertex. We ensure that by the end of the $i^{th}$ phase, each two neighboring clusters have identifiers that agree in the $i$ least significant bits. In the $(i+1)^{th}$ phase, clusters are categorized into \emph{red} or \emph{blue} clusters, based on the $(i+1)^{th}$ least significant bit (while all clusters of each connected component agree on the $i$ least significant bits, by the construction's induction). Then, we \emph{release} red clusters: their vertices might join one of the neighboring blue clusters, die, or remain in this red cluster if they have no neighboring blue cluster. On the other hand, each blue cluster retains all of its vertices and can also grow by accepting some of neighboring red vertices. This growth happens step by step, and hop by hop. 
Per step, each red node arbitrarily chooses a neighboring blue cluster to join, and each blue cluster checks the number of directly neighboring red vertices that want to join it. If they are at least a $1/(2b)$ fraction of the size of this blue cluster, they are accepted to join and they become blue. In this case, the cluster grows considerably in size, but also at most one hop in radius. But we cannot have more than $O(b\log n)$ such growth steps; beyond that the cluster would have more than $n$ vertices. On the other hand, if the fraction is less than a $1/(2b)$ fraction of the size of the blue cluster, all those red vertices die, and this blue cluster stops its growth for this iteration. This way, at the end of the steps of this phase, no edge remains between a blue and a red cluster, and at most a $1/(2b)$ fraction of all vertices die during the phase. 

At the end of $b$ phases, one for each bit in the identifiers, at most a $b/(2b)=1/2$ fraction of the vertices died, while each connected component of living vertices agrees on all the $b$ bits of the cluster identifier, i.e., is just one cluster. Since each cluster grows by at most one hop per each step of each phase, the cluster radii remain in $\poly(\log n)$. 

\subsection{Other Related Work}
We obtained the results in this paper after the second-named author learned about the statement of the main result of Kowalski and Krysta~\cite{Kowlaski-Krysta}, which claims to provide a $\poly(\log n)$ round algorithm for the splitting problem\footnote{Concretely, the second-named author received a request from the Program Committee of SODA 2020 to write a review on ~\cite{Kowlaski-Krysta}. That review request was declined due to the conflict of interest.}. Due to results of Ghaffari et al.\cite{ghaffari2017complexity}, this statement would imply an alternative proof for $\mathsf{P}\textit{-}\mathsf{LOCAL} = \mathsf{P}\textit{-}\mathsf{RLOCAL}$. Hence, that would be effectively equivalent to the main result in our paper (modulo aspects such as the exact polynomial in the round complexity, message size, local computational complexity, simplicity, etc). However, two remarks are in order: (1) The proof in \cite{Kowlaski-Krysta} has a fundamental flaw\footnote{Here is a brief explanation: In page 11 of ~\cite{Kowlaski-Krysta} (the version from 31 July 2019), the inequality $Pr[A_2|A_1] \leq 2^{-c'\alpha \delta}$ is incorrect. The provided argument says that this is derived by using the same arguments as done for $Pr[A_1]$ in Lemma 2.  However, that argument cannot be applied in the second phase and further on. For the second phase, we have hyperedges left each with potentially only $\alpha \delta$ vertices that are left uncolored. Notice that this is much lower than the $\delta$ vertices that was assumed when proving Lemma 2. Hence, recalling that an edge is called biased if it has at most $\alpha\delta$ red edges or at most $\alpha\delta$ blue edges, the probability of an edge being biased in the new coloring (i.e., second phase) is effectively $1$. If we change the definition of biased edges and allow the bias to go down by a constant factor per phase, this issue would be resolved superficially but then we can continue the argument for only $O(\log\log n)$ phases, as after that the hyperedges that initially had $\poly(\log n)$ vertices might have no uncolored vertex left.}. As of the time of preparing this version of our paper (i.e., \today), we are not aware of any fix to that proof. (2) The methods in the two papers are completely different, in terms of both the general approach and the proof ingredients.

%\newpage
\section{The Network Decomposition Algorithm}
\label{sec:algorithm}
In this section, we present a network decomposition algorithm that proves \Cref{thm:decomp-informal}. We first describe in \Cref{subsec:weak} an $O(\log^7 n)$-round deterministic distributed algorithm in the \local model that computes a weak-diameter network decomposition for $n$-node graphs, with cluster weak-diameter $O(\log^3 n)$ and $O(\log n)$ colors. This algorithm can also be adapted to work in $O(\log^8 n)$ rounds of the \congest model. Then, in \Cref{subsec:strong}, we explain how the former can be transformed to an $O(\log^8 n)$-time deterministic algorithm in the \local model for strong-diameter network decomposition, with cluster strong-diameter $O(\log n)$ and $O(\log n)$ colors. The distinction between weak-diameter and strong diameter is clarified in \Cref{subsec:weak}.

As a side remark, we note that all these constructions assume that nodes have unique $O(\log n)$-bit identifiers. As we will explain later in \Cref{rem:big_labels}), in the \local model, these constructions can be turned into $\poly(\log n)$-round algorithms for the more general setting with identifiers from $[1, S]$, as long as $\log^* S= O(\log n)$.

\subsection{Weak-Diameter Network Decomposition}
\label{subsec:weak}

Recall that for \Cref{thm:decomp-informal}, we wish to construct a decomposition of the underlying graph in $O(\log n)$ color classes such that for each color class, each of its connected components has $O(\log n)$ diameter. 
Our initial algorithm will, however, provide only a weaker property, as we describe next.
%that all of the vertices of a given connected component are close in the original graph $G$. 
We will work with \textit{clusters} of vertices, defined simply as a subset of vertices, such that any two vertices of a cluster are ``close" in $G$, although the subgraph induced by the vertices of the cluster may have large diameter and may be even disconnected. 
This motivates the notion of weak-diameter and the corresponding relaxation of network decomposition: 

\begin{definition}
\label[definition]{def:weak_diam}
Given a graph $G$ and its subgraph $H$, we say that the weak-diameter of $H$ is at most $d$ if $G$ contains a path of length at most $d$ between any pair of vertices in $H$. 
\end{definition}

\begin{definition}Given a graph $G$, we define a weak-diameter network decomposition of $G$ with $c$ colors and weak-diameter $d$ to be a coloring of the vertices with $c$ colors such that for each color $i\in [1, c]$, the subgraph $G_i$ induced by the vertices of color $i$ is partitioned into non-adjacent disjoint clusters, each of weak-diameter at most $d$ in graph $G$. 
\end{definition}

Next we state the main technical contribution of this paper, which is a deterministic distributed algorithm that constructs a weak-diameter decomposition in $\poly(\log(n))$ rounds in the \local model. With the known connection that transforms it to a strong-diameter decomposition algorithm, as we will later describe in \Cref{subsec:strong}, this implies \Cref{thm:decomp-informal}.

Before stating the result, we recall another useful notion of Steiner trees. A Steiner trees is a tree with nodes labelled as \textit{terminal} and \textit{nonterminal}; the aim is to connect terminal nodes possibly via some nonterminal nodes.
Here we use this notion to control the weak-diameter of each cluster.%, which is up to a constant factor bounded by the radius of its Steiner tree, hence each cluster has weak-diameter $O(\log ^3 n)$. 

\medskip

\begin{theorem}\label{thm:main} Consider an arbitrary $n$-node network graph $G$ where each node has a unique $b=O(\log n)$-bit identifier. 
There is a deterministic distributed algorithm that computes a network decomposition $G$ with $O(\log n)$ colors and weak-diameter $O(\log^3 n)$, in $O(\log^7 n)$ rounds of the \local model.

Moreover, for each color and each cluster $\mathcal{C}$ of vertices with this color, we have a Steiner tree $T_\mathcal{C}$ with radius $O(\log^3 n)$ in $G$, for which the set of terminal nodes is equal to $\mathcal{C}$. 
Furthermore, each edge in $G$ is in $O(\log^2 n)$ of these Steiner trees.
\end{theorem}

\medskip

The last part of the statement ensures that our algorithm can also be implemented and used in the more restrictive \congest model, as we will later discuss in \Cref{rem:main_congest}. 

In the following lemma, we describe the process for constructing the clusters of one color of the network decomposition (e.g., the first color), in a way that it clusters at least half of the vertices. This last weakening of the guarantee is similar to the randomized network decomposition algorithm of \cite{linial93}. Since after each application of this lemma only half of the vertices remain, by $\log n$ repetitions, we get a decomposition of all vertices, with $\log n$ colors.

%In the end, we will also observe that thanks to the fact that each edge of $G$ is in only $O(\log^2 n)$ Steiner trees, and because the construction only uses communication along Steiner trees, the algorithm can be run in $O(\log^8 n)$ rounds of the \congest model.

\begin{lemma}\label[lemma]{lem:main} Consider an arbitrary $n$-node network graph $G = (V, E)$ where each node has a unique $b=O(\log n)$-bit identifier, as well as a set $S\subseteq V$ of living vertices. There is a deterministic distributed algorithm that, in $O(\log^6 n)$ rounds in the \local model, finds a subset $S' \subseteq S$ of living vertices, where $|S'|\geq |S|/2$, such that the subgraph $G[S']$ induced by set $S'$ is partitioned into non-adjacent disjoint clusters, each of weak-diameter $O(\log^3 n)$ in graph $G$.

Moreover, for each such cluster $\mathcal{C}$, we have a Steiner tree $T_\mathcal{C}$ with radius $O(\log^3 n)$ in $G$ where all nodes of $\mathcal{C}$ are exactly the terminal nodes of $T_\mathcal{C}$. Furthermore, each edge in $G$ is in $O(\log n)$ of these Steiner trees.
\end{lemma}

We obtain \Cref{thm:main} by $c=\log n$ iterations of applying \Cref{lem:main}, starting from $S=V$. For each iteration $j\in [1, \log n]$, the set $S'$ are exactly nodes of color $j$ in the network decomposition, and we then continue to the next iteration by setting $S\gets S\setminus S'$.

\paragraph{Construction outline for one color of the decomposition}
We now describe the construction outline of \Cref{lem:main}.
The construction has $b=O(\log n)$ phases, corresponding to the number of bits in the identifiers. Initially, we think of all nodes of $S$ as \emph{living}. During this construction, some living nodes \emph{die}. We use $S'_i$ to denote the set of living vertices at the beginning of phase $i\in [0, b-1]$. Slightly abusing the notation, we let $S'_b$ denote the set of living vertices at the end of phase $b-1$ and define $S'$ to be the final set of living nodes, i.e., $S' := S'_{b}$. 

Moreover, we \emph{label} each living node $v$ with a $b$-bit string $\ell(v)$, and we use these labels to define the clusters. At the beginning of the first phase, $\ell(v)$ is simply the unique identifier of node $v$. This label can change over time. For each $b$-bit label $L\in \{0,1\}^b$, we define the corresponding \emph{cluster} $S'_i(L)\subseteq S'_i$ in phase $i$ to be the set of all living vertices $v\in S'_i$ such that $\ell(v)=L$. 
We will maintain one Steiner tree $T_L$ for each cluster $S'_i(L)$ where all nodes $S'_i(L)$ are the terminal nodes of $T_L$. 
Initially, each cluster consists of only one vertex and this is also the only (terminal) node of its respective Steiner tree. 
%each Steiner tree consists of only one node which is also terminal -- it is the only node in the respective cluster. 

\medskip
\paragraph{Construction invariants} The construction is such that, for each phase $i\in [0, b-1]$, we maintain the following invariants: 
\begin{enumerate}
\item[(I)] For each $i$-bit string $Y$, the set $S'_i(Y)\subseteq S'_i$ of all living nodes whose label ends in suffix $Y$ has no edge to other living nodes $S'_i\setminus S'_i(Y)$. In other words,
the set $S'_i(Y)$ is a union of some connected components of the subgraph $G[S'_i]$ induced by living nodes $S'_i$.
\item[(II)] For each label $L$ and the corresponding cluster $S'_i(L)$, the related Steiner tree $T_L$ has radius at most $iR$, where $R=O(\log^2 n)$. We emphasize that in the subgraph induced by living vertices a cluster can be disconnected.
\item[(III)] We have $|S'_{i+1}|\geq |S'_i|(1-1/2b)$. 
\end{enumerate}

These invariants, together with \Cref{obs:congestion} about the overlaps of the Steiner trees, prove \Cref{lem:main}. % (modulo the number of Steiner trees that use an edge, which will be discussed later in \Cref{rmrk:congestion}). 
In particular, from the first invariant we conclude that at the end of $b$ phases, different clusters are non-adjacent. 
From the second invariant we conclude that each cluster has a Steiner tree with radius $bR=O(\log^3 n)$. 
Finally, from the third invariant we conclude that for the final set of living nodes $S'= S'_{b}$, we have $|S'| \geq (1-1/2b)^{b} |S| \geq |S|/2$.

\medskip
\paragraph{Outline of one phase of construction} We now outline the construction of one phase and describe its goal (see also \Cref{fig:one_phase}). Let us think about some fixed phase $i$. We focus on one specific $i$-bit suffix $Y$ and the respective set $S'_i(Y)$. 
%Notice that by the invariant, the labels of all nodes in this component share a common suffix of $i$-bits $Y_{\mathcal{C}}$. 
Let us categorize the nodes in $S'_i(Y)$ into two groups of \emph{blue} and \emph{red}, based on whether the $(i+1)^{th}$ least significant bit of their label is $0$ or $1$. Hence, all blue nodes have labels of the form $(*\ldots*0Y)$ and all red nodes have labels of the form $(*\ldots*1Y)$, where $*$ can be an arbitrary bit. During this phase, we make some small number of the red vertices die and we change the labels of some of the other red vertices to blue labels (and then the node is also colored blue). All blue nodes remain living and keep their label. The eventual goal is that, at the end of the phase, among the living nodes, there is no edge from a blue node to a red node. Hence, each connected component of the living nodes consists either entirely of blue nodes or entirely of red nodes. Therefore, the length of the common suffix in each connected component is incremented, which leads to invariant 
(I) for the next phase. The construction ensures that we kill at most $|{S'}_i(Y)|/{2b}$ red vertices of set $S'_i(Y)$, during this phase. We next describe this construction.

\medskip
\paragraph{Steps of one phase}
Each phase consists of $R = 10b\log n = O(\log^2 n)$ steps, each of which will be implemented in $O(\log^3 n)$ rounds. Hence, the overall round complexity of one phase
%\footnote{We think that one should be able to implement a phase in $O(\log^3 n)$ rounds of the \local model, instead of $O(\log^5 n)$ rounds, by gathering some relevant local topology and then simulating the process locally. However, we leave such optimizations to a later version of this paper.} 
is $O(\log^5 n)$ and over all the $O(\log n)$ phases, the round complexity of the whole construction of \Cref{lem:main} is $O(\log^6 n)$ as advertised in its statement. 
Each step of the phase works as follows: each red node sends a request to an arbitrary neighboring blue cluster, if there is one, to join that blue cluster (by adopting the label). For each blue cluster $A$, we have two possibilities: 
\begin{enumerate}
\item[(1)] If the number of adjacent red nodes that requested to join $A$ is less than or equal to $|A|/2b$, then $A$ does not accept any of them and all these requesting red nodes die (because of their request being denied by $A$). In that case, cluster $A$ \emph{stops} for this whole phase and does not participate in any of the remaining steps of this phase. 
\item[(2)] Otherwise --- i.e., if the number of adjacent red nodes that requested to join $A$ is strictly greater than $|A|/2b$ --- then $A$ accepts all these requests and each of these red nodes change their label to the blue label that is common among all nodes of $A$. In this case, we also grow the Steiner tree of cluster $A$ by one hop to include all these newly joined nodes.
\end{enumerate}
We note that each step can be performed in $O(\log^3 n)$ rounds, because each blue cluster has a Steiner tree of depth $O(\log^3 n)$ and therefore can gather the number of vertices in the cluster, as well as the number of red vertices that would like to join this cluster. We also emphasize that in each step, each red node acts alone, independent of other nodes in the same red cluster. Hence, red clusters may shrink, disconnect, or even get dissolved over time. Once a red node adopts a blue label (or if a node had a blue label at the beginning), it will maintain that label throughout the phase. Therefore, blue clusters can only grow, and have more and more red nodes join them. We also emphasize that we can have blue clusters adjacent to each other, and red clusters adjacent to each other -- the objective is to have no edge connecting a red cluster to a blue cluster. 

Let us observe how the Steiner trees of the clusters evolve: For each blue cluster, the corresponding Steiner tree only grows. 
To have a similar property about the Steiner trees of red clusters, we do the following: Although for a red cluster, a terminal red node might become blue, we keep it in this tree as a nonterminal node. We note that although the Steiner tree of a red cluster is not used in the current phase, it may be used in the next phases.  

\begin{figure*}[t]
    \centering
    \includegraphics[width = \textwidth]{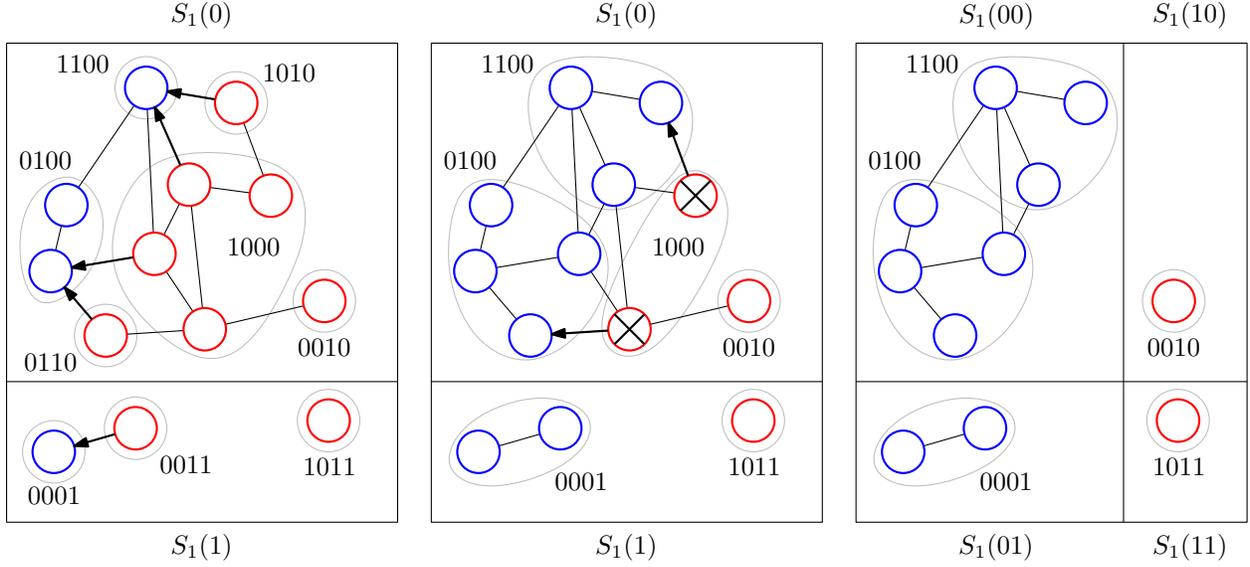}
    \caption{\small In this illustration, we consider the second phase of the algorithm, in a simple example graph. The three figures show the configuration in the beginning of three steps of this phase from left to right. Note that, at the beginning of this phase, the clusters are already separated according to their least significant bit (as a result of the first phase). 
    When the second phase starts---i.e., in the left figure---the second least significant bit determines whether each cluster is blue or red. 
    Adjacent red nodes are proposing to blue nodes (dark arrows) to join their clusters and blue clusters decide either to relabel them so that they join this cluster or to make them die (crossed red vertices). 
    In the end, blue and red clusters are separated. 
    Note that nothing will happen in the third phase, since the only two adjacent clusters share the same bit on the third least significant bit. Their boundary will be resolved only in the last phase. 
    }
    \label{fig:one_phase}
\end{figure*}

\medskip
\paragraph{Analysis} We next provide some simple observations about this construction in one phase, which allow us to argue that the construction maintains invariants (I) to (III), described above.

\begin{observation}\label[observation]{obs:finite} Any blue cluster stops after at most $4b\log n$ steps. 
\end{observation}
\begin{proof} In each step that a cluster $A$ does not stop, its size grows by a factor of at least $(1+1/2b)$, as it accepts at least $|A|/2b$ requests from neighboring red nodes. Hence, after $4b\log n$ steps of growth, the size would exceed $(1+1/2b)^{4b\log n} > n$, which is not possible. Therefore, cluster $A$ stops after at most $4b\log n$ steps.
\end{proof}

\begin{observation}\label[observation]{obs:stop} Once a blue cluster $A$ stops, it has no edge to a red node and it will never have one, during this phase. This implies invariant (I).
\end{observation}
\begin{proof} By the observation above, cluster $A$ stops after at most $4b\log n$ steps. Consider the step in which cluster $A$ stops. In that step, each neighboring red node (if there is one) either requested to join $A$ or some other blue cluster. In the former case, that red node dies. In the latter case, the node adopts a blue label or dies. In either case, the node is not a living red node anymore (and it will never become one). From this point onward, this blue cluster $A$ never grows or shrinks. 
\end{proof}

\begin{observation}\label[observation]{obs:inv_two} In each step, the radius of the Steiner tree of each blue cluster grows by at most $1$, while the radius of the Steiner tree of each red cluster does not grow. This implies invariant (II).
\end{observation}

\begin{observation}\label[observation]{obs:inv_three} The total number of red vertices in $S'_i(Y)$ that die during this phase is at most $|S'_i(Y)|/(2b)$. This implies invariant (III).
\end{observation}
\begin{proof} From \Cref{obs:stop}, it follows that each blue cluster $A$ stops exactly once, and if it had $|A|$ vertices at that point, it makes at most $|A|/(2b)$ red vertices die. Hence, in total over the whole phase, the number of red vertices that die is at most a $1/(2b)$ fraction of the number of nodes in blue clusters that stop, and thus at most $|S'_i(Y)|/(2b)$.
\end{proof}

The above completes the description of our algorithm in the \local model. 
As we will later remark about its applications in the \congest model, we finish the proof of \Cref{lem:main} by adding the following observation about the overlaps of the constructed Steiner trees. 

\begin{observation}\label[observation]{obs:congestion} Eeach edge is used in $O(\log n)$ Steiner trees.
\end{observation}
\begin{proof}
Each edge can be in the Steiner tree of a cluster only if that cluster at some point included one of the two endpoints of this edge. Throughout the construction, each node changes its label at most $b = O(\log n)$ times, i.e., at most once per label bit. Hence, each edge is used in $O(\log n)$ Steiner trees.
\end{proof}

\medskip 

Below, just to help with the intuition, we discuss an idealized global view of the process in one phase. We then state some remarks about extensions of the result to the \congest model and the settings with larger identifiers.

\paragraph{An intuitive global view of one (ideal) phase} We next describe a different global view for an idealized version of the process in one phase. 
We hope that this view helps in understanding the process; concretely, the above process can be seen as a \emph{localized} version of the idealized global view, where some decisions are performed locally (thus, the colors of nodes might differ in the two processes, but the growth of blue nodes and the number of red nodes that die, when the growth stops, behave similarly). 

The process described above for one phase intends to make sure that there is no edge between red and blue living nodes, while the number of (red) nodes that die is kept small. For that, we grow the blue clusters locally (i.e., relabeling some red nodes to adopt blue labels, while keeping each blue label for the entire phase), each by $O(\log^2 n)$ hops, while making some red nodes die in the meantime. The process also guarantees that only a $1/2b$ fraction of nodes die. 
If we were to ignore the exact labels of the blue nodes and red nodes, and we would just remember whether a node is red or blue, the quantitative aspects of this process --- namely the number of steps of growth and the number of red nodes that die --- would resemble a simpler global ball carving argument: we would start from the initial ``ball" of all blue nodes being together, and would grow this ``ball" hop by hop, as long as in each step we grow by at least a $(1+1/2b)$ factor. In the first step that there is no such rapid growth --- which will happen within $4b\log n$ steps --- we would carve all the neighboring red nodes and call them dead. That would be at most a $1/2b$ fraction of the blue nodes (and hence all nodes). Once these boundary nodes are dead, there is no edge between living red and blue nodes.

\paragraph{Remark about the length of identifiers}
For the construction in the \local model, the requirement on the size of the identifiers of each node can be substantially weakened; this is important for applications when we use the algorithm in the shattering framework, e.g., \cite{gmis,barenboim_symmbreaking}.

\begin{remark}\label[remark]{rem:big_labels} In the construction provided above, we assumed that the nodes of the $n$-node graph have $O(\log n)$-bit unique identifiers. This construction can be extended to an $O(\log^* S \cdot \log^6 n)$-round algorithm in the \local model for the setting where the identifiers are from $[1, S]$.
\end{remark}
\begin{proof}
Let $T(n):=O(\log^7 n)$ be the complexity of the algorithm in $n$-node graphs with $(3\log n)$-bit labels. We compute an $O(n^2)$ coloring of $G^{T(n)}$---the graph on the same set of vertices as $G$ but where we connect every two vertices $v$ and $u$ that have distance at most $T(n)$ in $G$---using the coloring algorithm of Linial\cite{linial92}, in $T(n)\cdot \log^*(S) = O(\log^7 n) \cdot \log^*(S)$ rounds. We recall that Linial's algorithm provides a $O(\Delta^2)$-coloring of any graph with maximum degree at most $\Delta$ where nodes have identifiers from $[1, S]$ in $O(\log^*S)$ rounds of the \local model.
Once we compute a coloring of $G^{T(n)}$ with $O(n^2)$ colors, we can then adopt these colors as \lq\lq unique\rq\rq\ identifiers with no more than $(3\log n)$-bits. Since each node sees unique identifiers in its $(T(n))$-hop neighborhood, the \local algorithm works as if nodes had unique identifiers.
\end{proof}

\subsection{Construction in the CONGEST Model and Extension to Graph Powers}
\label{subsec:CONGEST-construction}

Although we formulated the algorithm in the \local model of computation, it can be easily observed that it also runs in the more restrictive \congest model. 
\begin{remark} 
\label[remark]{rem:main_congest}
The whole network decomposition construction described in \Cref{lem:main} can be performed in $O(\log^8 n)$ rounds of the \congest model.
\end{remark}
\begin{proof} Recall from above that in the construction of clusters of one color, each edge is used in $O(\log n)$ Steiner trees. 
Moreover, whenever we add an edge to a particular Steiner tree, we can think of it as being oriented from a newly added node towards a node that was already in the tree. 
This gives a natural orientation of its edges that points to its root, which is the vertex whose original identifier is equal to the label of the cluster, and that was initially the only member of this cluster. 

The construction only uses convergecast and broadcast along these rooted trees (to decide whether the cluster should continue growing or it should stop). 
Hence, by using every $O(\log n)$ rounds of \congest model as one \emph{big-round}, we can perform the construction of one color in $O(\log^6 n)$ big-rounds, i.e., $O(\log^7 n)$ rounds of the \congest model. 
Over all the $O(\log n)$ colors, this leads to a round complexity of $O(\log^8 n)$ rounds of the \congest model.
\end{proof}

In the \congest model it is particularly helpful that \Cref{lem:main} gives us the underlying Steiner tree for each cluster, with the property that each appears in only $O(\log n)$ trees per color. These Steiner trees can later be used for simultaneous broadcast or convergecast in the clusters. 

\paragraph{Extending the CONGEST-model construction to graph powers}
%\marginpar{remark about the application; shattering}
When solving a distributed problem for an underlying graph $G$, it is often helpful to simulate its power $G^k$ and run certain algorithm on this simulated graph (for example, this will be the case in \Cref{thm:main-strong} as well as all the applications in \Cref{subsub:congestDerand}). 
Obtaining a network decomposition for $G^k$ is straightforward in the \local model, where each node can start by collecting its $k$-hop neighbourhood and then simulate each step of an algorithm for $G^k$ with additional slowdown proportional to $k$. 
However, this cannot be done easily in the \congest model\footnote{Even the task of collecting $2$-hop neighbourhood of a given node $u$ cannot be generally solved in $\poly (\log n)$ rounds, since the number of vertices in the $2$-hop neighbourhood of $u$ can be much larger than the number of connections of $u$ to its neighbours that can be used to collect information.}. 
That being said, our algorithm can be adapted to provide a weak-diameter network decomposition for $G^k$ in the \congest model without the need of an explicit construction of $G^k$.

A weak-diameter decomposition of $G^k$ with $c$ color classes of weak-diameter $d$ can be also interpreted as a weak-diameter decomposition of $G$ with $c$ color classes of weak-diameter $k\cdot d$, where any two clusters of the same color class are at least $k+1$ hops apart. 
%The number of rounds needed will be again multiplied by $k$. 

\begin{theorem}
\label[theorem]{thm:main_alg_gk}
There is an algorithm in the \congest model that, given a value $k$ that is known to all nodes, in $O(k \log^8 n \cdot \min(k,\log^2 n))$ communication rounds outputs a weak-diameter network decomposition of $G^k$ with $O(\log n)$ color classes, each one with $O(k \cdot \log^3 n)$ weak-diameter in $G$. 

Moreover, for each color and each cluster $\mathcal{C}$ of vertices with this color, we have a Steiner tree $T_\mathcal{C}$ with radius $O(k \cdot \log^3 n)$ in $G$, for which the set of terminal nodes is equal to $\mathcal{C}$. 
Furthermore, each edge in $G$ is in $O(\log^2 n \cdot \min(k, \log^2n))$ of these Steiner trees.
\end{theorem}

The proof idea is to run the algorithm from \Cref{lem:main}, with one change: vertices that will propose in one step will not be just red vertices bordering with a blue vertex, but all red vertices in a $k$-hop neighbourhood of some blue vertex. 
This idea by itself readily gives us a $\poly(\log n)$ round algorithm for $k=\poly(\log n)$. As we will show, with some more work, we can also get an algorithm with the same round complexity as the algorithm from \Cref{rem:main_congest} whenever $k$ is constant. Below, we provide a concrete proof sketch. \shortOnly{A full proof can be found in the full version of this paper on arXiv.} \fullOnly{The full proof is deferred to \Cref{app:congest-powergraph-decomp}.}

\begin{proof}[Proof sketch of \Cref{thm:main_alg_gk}]

Consider the algorithm from \Cref{thm:main} with the following change. We generalize each \emph{step} --- red nodes with a blue neighbour are proposing to an arbitrary blue neighbour--- so that all the nodes with at least one blue neighbour in the their $k$-hop neighbourhood are proposing. 
We can mark all such red nodes in $k$ rounds by running a breadth first search (BFS) from all blue nodes simultaneously. 
This can be implemented in such a way that each edge is used in at most one round, hence it can be done in $k$ rounds of the \congest model.  
Note that running BFS also naturally outputs an oriented forest with blue nodes as the roots. 
This forest will also contain dead nodes. 
Now it is simple to implement proposals of red nodes to blue clusters: each red node in a forest will propose to its root and, if accepted by the respective blue cluster, the whole path to the root (that potentially also contain dead nodes) is added to the corresponding Steiner tree. 

It is easy to see that this adaptation of our algorithm returns a network decomposition of $G^k$. 

Before we analyze the running time, let us add an additional optimization to ensure that each edge will be added to at most $k$ Steiner trees during one phase (it is added to at most one Steiner tree during one phase in the original algorithm):
During a phase and in between steps, a dead node $v$ that was used in a BFS tree rooted at $r$ will change to a different BFS tree only if the blue root $r'$ of the new BFS tree is closer to the node $v$ than its current root $r$. 
This property will hold in our implementation of the algorithm because of the following two properties: (A) we break symmetry during BFS in the same manner in all the steps (e.g., toward the one with smaller ID) and (B) even blue clusters that stopped growing are running BFS in every step, to make sure that a dead node changes a BFS tree only when a closer blue cluster appeared. With this optimization, each edge is used by at most $O(\min(k, \log^2(n)))$ different Steiner trees in a phase and, hence, by at most $O(\log^2(n) \cdot \min(k, \log^2(n)))$ Steiner trees in total. 

We are now ready to bound the running time. 
The algorithm constructs $O(\log n)$ color classes and each one is constructed in $O(\log n)$ phases with each phase containing $O(\log^2 n)$ steps. 
For each step, we need to run breadth first search for $O(k)$ steps and broadcast information to root via Steiner trees of depth $O(k \log^3 n)$, which dominates the $O(k)$ steps of the breadth first search. 
Moreover, each edge is used by $O(\log n \cdot \min(k + \log^2 n))$ of the Steiner trees. 
This implies a round complexity $O(k \log^8 n \cdot \min(k + \log^2 n))$. 
\end{proof}

%Although a simple extension of the previous algorithm from \Cref{lem:main}, we now need to take more care about the underlying Steiner trees and we need to rethink invariants from \Cref{lem:main}. 

\subsection{Strong-Diameter Network Decomposition}
\label{subsec:strong}

We now explain that by a known method, first presented by Awerbuch et al.\cite{awerbuch96}, in the \local model, we can transform the algorithm of \Cref{thm:main} for weak-diameter network decomposition to an algorithm for strong-diameter network decomposition, which thus proves \Cref{thm:decomp-informal}. Since this is a known connection, we provide only a sketch of the proof.
\begin{definition}Given a graph $G=(V, E)$, we define a network decomposition of $G$ with $c$ colors and strong-diameter $d$ to be a 
partitioning of the vertices into $c$ disjoint color classes $V_1$, \dots, $V_{c}$, such that in the subgraph $G[V_i]$ induced by the vertices of each color $i$, each connected component has diameter at most $d$. Each of these connected component of the subgraph $G[V_i]$ is called a cluster. 
\end{definition}
The following theorem statement is a rephrased and formalized version of \Cref{thm:decomp-informal}:
\begin{theorem}\label{thm:main-strong} Consider an arbitrary $n$-node network graph where each node has a unique $b=O(\log n)$-bit identifier. There is a deterministic distributed algorithm that computes a network decomposition of this graph with $O(\log n)$ colors and strong-diameter $O(\log n)$, in $O(\log^8 n)$ rounds in the \local model. 
\end{theorem}
\begin{proof}[Proof Sketch] 
We first recall the standard sequential algorithm for building a network decomposition with $\log n$ colors and strong-diameter $O(\log n)$ per cluster, and then we explain how the weak-diameter network decomposition algorithm of \Cref{thm:main} helps us build such a strong-diameter decomposition in a distributed manner. 

The standard sequential algorithm algorithm for building a network decomposition with $\log n$ colors and strong-diameter $O(\log n)$ per cluster works as follows. We describe the process for determining the nodes of the first color. The other colors are obtained similarly, by applying the same construction repeatedly for $\log n$ times, each to the graph induced by the remaining nodes. For the first color, starting with an arbitrary node $v$, we do a \emph{sequential ball carving}. That is, we grow a \emph{ball} around this vertex, hop by hop. A ball of radius $r$ is simply all nodes that are within distance $r$ of node $v$. We increment the radius $r$ of this ball gradually, one by one, as long as the number of the nodes outside the ball that are adjacent to the ball is at least equal to the number of nodes in the ball. Once this growth condition is not satisfied, which will happen before $r\leq \log n$ as otherwise the ball has more than $n$ nodes, we consider all nodes in the ball as one cluster (in the first color) of the decomposition, and we consider all nodes outside but adjacent to it as \emph{dead}. All dead nodes are removed from the construction of the first color of the network decomposition. Then, repeatedly, if any node remains that is not dead but also not clustered, we continue a similar sequential ball carving process starting from that node, among nodes that are not dead. This gives the clusters of the first color of the decomposition. Then, we bring all the dead nodes back to life and repeat this process among them, getting the clusters of the second color, and so on. Since each time we cluster at least $1/2$ of the vertices, we finish after at most $\log n$ repetitions, i.e., at most $\log n$ colors.

We now explain how  \Cref{thm:main} allows us to get an efficient distributed simulation of this sequential construction, thus proving \Cref{thm:main-strong}. Let $G':=G^{10\log n}$, i.e., the graph on the same set of vertices as $G$ but where we connect every two vertices $v$ and $u$ that have distance at most $10\log n$ in $G$. We apply the algorithm of \Cref{thm:main} to obtain a weak-diameter network decomposition of $G'$, in $O(\log^8 n)$ rounds of communication on $G$. The resulting network decomposition is a coloring of vertices with $q=O(\log n)$ colors where the clusters in each color have weak-diameter $O(\log^3 n)$ in $G'$, and thus weak-diameter $O(\log^4 n)$ in $G$. We next use this \emph{helper} network decomposition to build our \emph{output} strong-diameter network decomposition with $O(\log n)$ colors and $O(\log n)$ diameter. 

We describe the process for determining the nodes of the first color in the output network decomposition. The other colors are obtained similarly, by applying the same construction repeatedly for $O(\log n)$ times, each to the graph induced by the remaining nodes. 

To determine the nodes of the first color of the output decomposition, we process the colors of the helper network decomposition one by one, in $q$ stages. Let us fix one stage (and thus one color of the helper network decomposition, and its clusters). For each cluster, we elect a leader for it and we gather the topology of the subgraph of all remaining nodes within $\log n$ hops of the nodes of this cluster. Notice that since the cluster has weak-diameter $O(\log^4 n)$, this can be done in $O(\log^4 n)$ rounds. Moreover, the topologies gathered by different clusters are disjoint. 
This is because different clusters of this color of the helper decomposition have distance at least $10\log n$, since otherwise $G'$ would contain an edge connecting the two clusters.

Each cluster $\mathcal{C}$ will perform a sequential ball carving process, on the topology that it has gathered, as follows: We start from an arbitrary node $v$ of color $i\in [1, q]$ in cluster $\mathcal{C}$, and grow a \emph{ball} around it, hop by hop, in the subgraph induced by the remaining nodes. A ball is simply all remaining nodes that are within a certain distance of node $v$, in the subgraph induced by the remaining nodes. We grow the radius of this ball gradually, as long as the number of the nodes outside the ball that are adjacent to the ball is at least equal to the number of nodes in the ball. Once this growth condition is not satisfied, we consider all nodes in the ball as one cluster of the output decomposition, and we consider all nodes outside but adjacent to it as \emph{dead}. All dead nodes are removed from the construction of this color of the output network decomposition. Then, if any node $v'$ of  cluster $\mathcal{C}$ remains unclustered (for the output decomposition), we start a similar ball growing process from $v'$, but only on the graph induced by the remaining nodes. We continue similarly until all nodes of 
cluster $\mathcal{C}$ are clustered for the output decomposition. 

In each step of growing a ball, the number of nodes grows by a $2$ factor. Hence, any ball can grow by at most $\log n$ hops. This implies that the ball growing processes from cluster $\mathcal{C}$ will never reach the ball growing processes from any other cluster $\mathcal{C}'$ of color $i$ of the helper decomposition. Furthermore, each time that we stop a ball's growth, the number of nodes on the boundary of it that die is less than the number of nodes inside the ball (which get clustered for the output network decomposition). Hence, after going through all the $q$ stages, at least $1/2$ of living nodes get clustered, and at most $1/2$ of living nodes die.

Then, we bring all dead nodes back to life and proceed to build the next color of the output network decomposition, only on the subgraph induced by these remaining nodes. As per repetition the number of remaining nodes reduces by a $2$ factor, we finish in $\log n$ repetitions. \end{proof}

\section{Implications and Applications}
\label{sec:implications}
As mentioned before, despite its simplicity, our efficient deterministic network decomposition has far-reaching implications, leading to a general efficient distributed derandomization theorem and  better \emph{deterministic} and \emph{randomized} distributed algorithms for a range of problems, as well as some improvements in massively parallel computation (aka, the MapReduce algorithms).  We next overview these implications. We start in \Cref{subsec:MIS-and-coloring} with the well-studied problems of maximal independent set and coloring, which were among the most well-known open problems in distributed graph algorithms and get settled immediately by our network decomposition. This also serves as a warm up for the standard method of using network decomposition. Then, in \Cref{subsec:Derand}, we present our general derandomization result for the \local model, thus proving \Cref{thm:derandomization}. Finally, in \Cref{subsec:OtherImplications}, we overview a list of other well-studied problems for which we get substantial (deterministic or randomized) improvements.

\subsection{Maximal Independent Set and Coloring}
\label{subsec:MIS-and-coloring}
\subsubsection{MIS}
The Maximal Independent Set (MIS) problem is one of the central problems in the study of distributed graph algorithms. As mentioned before, there have been  well-known $O(\log n)$-round randomized algorithm for this problem since the 1980s\cite{luby86, alon86} but obtaining a deterministic algorithm for it had remained open. 

\paragraph{Deterministic MIS} We next explain how the efficient network decomposition of \Cref{thm:main} directly gives a $\poly(\log n)$-round deterministic MIS algorithm. This already answers Linial's long-standing open question and settles Open Problem 11.2 in the book of Barenboim and Elkin \cite{barenboimelkin_book}. Weaker forms of this problem appear as Open Problems 11.5 and 11.8 in the same book\cite{barenboimelkin_book} and they are now resolved. The method is fairly standard and thus we provide a proof sketch. It also allows us to recall the usual method of using network decomposition to solve problems such as maximal independent set and coloring \cite{awerbuch89}.

\begin{corollary}\label{thm:MIS-local} There is a deterministic distributed algorithm, in the \local model, that computes a maximal independent set in $\poly(\log n)$ rounds.
\end{corollary}

\begin{proof}
First, we compute a network decomposition with $O(\log n)$ colors and clusters of diameter $O(\log^3 n)$, in $O(\log^7 n)$ rounds, using \Cref{thm:main}. Then, we process the clusters color by color. In each color $i$, the center node of each cluster aggregates at the center the topology of the cluster as well as the information of which nodes adjacent to the cluster have already been added to the maximal independent set, when processing the previous colors $1$ to $i-1$. Since the cluster diameter is $O(\log^3 n)$, this information can be gathered in $O(\log^3 n)$ rounds. Then, the center simulates a greedy process of adding the vertices of this cluster to the MIS, one by one, for any node that does not already have a neighbor in the MIS. Since any two cluster of the same color are non-adjacent, the computations of different clusters can happen simultaneously. Processing each color takes $O(\log^3 n)$ rounds, which means that we finish processing all the $O(\log n)$ colors in $O(\log^4 n)$ rounds. Together with the $O(\log^7 n)$ rounds used for computing the network decomposition, this is a deterministic maximal independent set algorithm that runs in $O(\log^7 n)$ rounds. 
\end{proof}

We note that, due to a very recent breakthrough of Balliu et al.\cite{balliu2019LB}, any deterministic algorithm for MIS needs a round complexity of $\Omega(\log n/\log\log n)$. 

\paragraph{Randomized MIS} Plugging the above deterministic MIS algorithm into the shattering framework of the algorithm of \cite{gmis} improves also the randomized complexity of MIS:
\begin{corollary}
There is a randomized distributed algorithm, in the \local model, that computes a maximal independent set in $O(\log \Delta) + \poly(\log\log n)$ rounds, with probability at least $1-1/\poly(n)$.
\end{corollary}
We note that due to a celebrated lower bound of Kuhn, Moscibroda and Wattenhofer~\cite{kuhn16_jacm}, any (randomized) algorithm for MIS needs a round complexity of $\Omega(\frac{\log \Delta}{\log\log \Delta})$, which means the $\Delta$ dependency in the above algorithm is nearly optimal. 
Moreover, regarding the dependency on $n$, due to another result of Balliu  et al.\cite{balliu2019LB}, any randomized algorithm for MIS needs a round complexity of $\Omega(\frac{\log \log n}{\log\log\log n})$, on some graphs with $\Delta = \Omega(\frac{\log \log n}{\log\log\log n})$. Thus, one cannot hope for an algorithm with round complexity $O(\log \Delta) + o(\frac{\log \log n}{\log\log\log n})$, or even $o(\Delta) + o(\frac{\log \log n}{\log\log\log n})$.
\smallskip

\paragraph{MIS with small messages} The algorithm described in the proof of \Cref{thm:MIS-local} works in the \local model, where message sizes are unbounded. We can also obtain an algorithm for the \congest model, where message sizes are bounded to $O(\log n)$:
\begin{corollary}\label{thm:congest-MIS} There is a deterministic distributed algorithm, in the \congest model, that computes a maximal independent set in $\poly(\log n)$ rounds.
\end{corollary}
\begin{proof}[Proof Sketch] The method outline is similar to the \local model algorithm, with two exceptions: (1) we use the \congest-model variant of our network decomposition, which runs in $O(\log^8 n)$ rounds, (2) when processing each cluster, we use a \congest-model MIS algorithm of Censor-Hillel, Parter, and Shwartzman \cite{censor2017derandomizing}, instead of the naive topology gathering step. 
Concretely, Censor-Hillel et al. give an $O(D\log^2 n)$-round MIS algorithm in the \congest model, $D$ denotes the graph diameter. When processing the colors of network decomposition, for each cluster of the color, we can run the algorithm of Censor-Hillel et al. on the cluster (ignoring nodes that already have a neighbor in the MIS). 
Recall from \Cref{lem:main} that per color, each edge of the graph is used by the Steiner trees of $O(\log n)$ clusters. 
Hence, we can run the algorithm of Censor-Hillel et al. for all the clusters of the same color, in parallel, in $O(\log^3 n \cdot \log^2 n \cdot \log n) = O(\log^6 n) $ rounds. Over all the $O(\log n)$ colors, this MIS computation runs in $O(\log^7 n)$ rounds of the \congest model, besides the initial $O(\log^8 n)$ rounds spent for computing a network decomposition. 
\end{proof}

\subsubsection{Coloring} 
\paragraph{Deterministic coloring} One can apply the standard method for using network decomposition, as done above when proving \Cref{thm:MIS-local}, to also obtain an $O(\log^7 n)$ round algorithm for $\Delta+1$ vertex coloring, where $\Delta$ denotes the maximum degree, or its generalization to list-coloring. This efficient coloring resolves Open Problem 11.3 in the book of Barenboim and Elkin\cite{barenboimelkin_book} and gives an alternative, and more systematic, solution for Open Problem 11.4, which asked for an efficient deterministic $(2\Delta-1)$-edge coloring (that problem was settled first in \cite{FischerGK17}).

\begin{corollary}\label{thm:det-coloring-local} There is a deterministic distributed algorithm, in the \local model, that computes a $(\Delta + 1)$ vertex coloring, where $\Delta$ denotes the maximum degree in the graph, in $\poly(\log n)$ rounds. The algorithm can also be generalized to list-coloring where each vertex $v$ should choose its color from a list $L_v$ of colors, where $|L_v| \geq deg(v)+1$.
\end{corollary}

\paragraph{Randomized coloring} Moreover, plugging this deterministic list-coloring algorithm of \Cref{thm:det-coloring-local} into the randomized coloring algorithm of Chang, Li, and Pettie\cite{chang2018optimal} improves the randomized complexity of $\Delta+1$ coloring from $2^{O(\sqrt{\log\log n})}$ to $\poly(\log\log n)$:

\begin{corollary}\label{crl:rand-coloring-local} There is a randomized distributed algorithm, in the \local model, that computes a $(\Delta + 1)$ vertex coloring, where $\Delta$ denotes the maximum degree in the graph, in $\poly(\log \log n)$ rounds, with probability at least $1-1/\poly(n)$.
\end{corollary}
\begin{proof}[Proof Sketch] Following the shattering framework\cite{barenboim_symmbreaking}, the randomized phase of the algorithm of \cite{chang2018optimal} works in $O(\log^* \Delta)$ rounds, and colors almost all nodes, except for some small components of nodes that remain uncolored. The guarantee is that, with probability at least $1-1/\poly(n)$, each remaining component has $\poly(\log n)$ vertices. After that, for the deterministic phase, we can invoke the deterministic list-coloring algorithm of \Cref {thm:det-coloring-local} on each of these components separately, all in parallel. Since each component has $\poly(\log n)$ vertices, this would run in $\poly(\log (\poly(\log n))) = \poly(\log\log n)$ rounds, and would complete the partial coloring to a coloring for all vertices.
\end{proof}

As another coloring result, by using \Cref{thm:det-coloring-local} along with the method of \cite{barenboim2011deterministic}, one can obtain an arboricity-dependent coloring:

\begin{corollary}
There is a deterministic distributed algorithm that computes a $(2+o(1))a$-coloring of any graph with arboricity at most $a$, in $\poly(\log n)$ rounds of the \local model.
\end{corollary}

\paragraph{Massively Parallel Computation (MPC) of coloring} we also get a nearly-exponential improvement for massively parallel (aka, MapReduce) algorithms\cite{karloff2010mpc} for $\Delta+1$ coloring. It is beyond the scope of this paper to explain the exact setting and review the related literature. For those, and particularly for the coloring problem, we refer the readers to \cite{karloff2010mpc, chang2019coloring, ghaffari2019conditional}. We just briefly state that in the MPC model (with strongly sublinear memory per machine), the $n$-node graph is partitioned among a number of machines, each with memory $n^{\alpha}$ for a constant $\alpha<1$, and per round each machine can send $n^{\alpha}$ bits to the other machines.

We obtain our improvement by plugging in the \local-model deterministic list-coloring algorithm of \Cref{thm:det-coloring-local} into the algorithm of \cite{chang2019coloring}. This gives a randomized MPC $\Delta+1$ coloring algorithm, with strongly sublinear memory per machine, with round complexity of $O(\log\log\log n)$, which improves on the previous bound of $O(\sqrt{\log\log n})$.

\begin{corollary}
There is a randomized MPC algorithm, in the regime where each machine has memory $n^{\alpha}$ for any constant $\alpha<1$, that computes a $\Delta+1$ coloring of any $n$-node graph with maximum degree at most $\Delta$ in $O(\log\log \log n)$ rounds, with high probability.
\end{corollary}
We also note that due to a conditional hardness result of \cite{ghaffari2019conditional}, conditioned on a standard hardness assumption of $\Omega(\log n)$-complexity for connectivity, improving this $O(\log\log\log n)$-round randomized MPC coloring algorithm would imply a deterministic $\log^{o(1)} n$-round deterministic distributed algorithm for $\Delta+1$ coloring, in the \local model, which would be a major improvement on the state of the art (\Cref{thm:det-coloring-local}). 

\subsection{Derandomization via Network Decomposition}
\label{subsec:Derand}
We now explain how our network decomposition, when put together with the approach of \cite{ghaffari2018derandomizing, ghaffari2017complexity}, leads to an efficient derandomization method for the \local model. We note that this result can be viewed as answering Open Problem 11.1 in the book of Barenboim and Elkin \cite{barenboimelkin_book}, which asked for developing ``a general derandomization technique for the distributed message passing model" and was followed by several locally checkable problems that admit $\poly(\log n)$-round randomized algorithms but no known $\poly(\log n)$-round deterministic algorithm. 

\medskip
\noindent \textbf{Theorem 1.1 (LOCAL Derandomization Theorem)}\emph{ 
We have $$\mathsf{P}\textit{-}\mathsf{LOCAL} = \mathsf{P}\textit{-}\mathsf{RLOCAL}.$$ Here, $\mathsf{P}$-$\mathsf{LOCAL}$ denotes the family of {locally checkable problems} that can be solved by deterministic algorithms in $\poly(\log n)$ rounds of the \local model in $n$-node graphs and $\mathsf{P}$-$\mathsf{RLOCAL}$ denotes the family of \emph{locally checkable problems} that can be solved by randomized algorithms in $\poly(\log n)$ rounds of the \local model, with success probability $1-1/n$.}
\begin{proof}[Proof Sketch]
A formal and precise description of this procedure can be found in \cite{ghaffari2018derandomizing}. To keep this article self-contained and accessible to a broad audience, we provide a less formal sketch here, and without going through the language of the $\mathsf{SLOCAL}$ model of \cite{ghaffari2017complexity}. 

Consider any locally checkable problem $\mathcal{P}$ that can be checked in $t(n)$ rounds by a deterministic \local-model algorithm, and a randomized \local-model algorithm $\mathcal{A}$ for $\mathcal{P}$ that runs in exactly $r(n)$ rounds and produces correct outputs with probability at least $1-1/\poly(n)$. Thus, composing these, we have an algorithm $\mathcal{B}$ that runs in $R = r(n)+t(n)$ rounds and computes the outputs for $\mathcal{P}$, as well as a correctness indicator flag $f_v$ for each node $v$ such that if a constraint of $\mathcal{P}$ involving node $v$ is not satisfied, then $f_v =1$. In other words, if for all nodes $v\in V$ the indicator flags $f_v=0$, the output is a valid solution for the problem. Moreover, the expected number of flags that equal to $1$ is at most $1/\poly(n)$. 
We derandomize this algorithm $\mathcal{B}$ by working through the network decomposition, and fixing the randomness of different nodes, via a method of conditional expectation for the function $\sum_{v} f_v$. 

We first take a network decomposition of $G^{2R+1}$ where each two nodes are connected if their distance is at most $2R+1$. This can be computed deterministically in $R\poly(\log n)$ rounds of the \local model, using \Cref{thm:main}. We get a decomposition into clusters of radius $O(R \log^3 n)$, colored with $O(\log n)$ colors, such that any two clusters of the same color are more than $2R+1$ hops apart.

Then, similar to the standard method explained in the proof of \Cref{thm:MIS-local}, we work through the colors of the network decomposition, one by one. Per color $i$, each cluster gathers the topology from $2R$-hop neighborhood of the cluster in the cluster center (this topology also includes the information of how randomness has been fixed, when processing previous colors), in $O(R\log^3 n)$ rounds. Then, each cluster center fixes the randomness of its vertices one by one, in a sequential manner, ensuring that the expectation of $\sum_{v} f_v$ conditioned on the fixed randomness does not increase. Notice that since $\mathcal{B}$ is an $R$ round algorithm, the randomness of each node $u$ influences only $f_v$ for nodes $v$ that are within distance $R$ of node $u$. Hence, the cluster center can compute the change in the expected value of $\sum_{v} f_{v}$ when fixing the randomness of each node $u$ in its cluster, and can fix the randomness in a way that does not increase the conditional expectation. Moreover, clusters of the same color can work in parallel as they are more than $2R+1$ hops apart and hence they do not influence the same indicator flag $f_v$ for any node $v$. Once each cluster center fixes the randomness of the node's of its cluster, it reports these values back to the nodes, in $O(R\log^3 n)$ rounds.
Then, we proceed to the next color and repeat a similar procedure. Once we finish processing all the $O(\log n)$ colors, all the randomness is fixed, and still the expected value of $\sum_{v} f_{v}$ is at most $1/\poly(n) \ll 1$. Since $\sum_v f_v$ has to be a non-negative integer value, we must have $\sum_v f_v =0$, which means all $f_v=0$ and thus all the constraints are satisfied. Overall, we now have a deterministic algorithm that runs in $R \cdot \poly(\log n)$ rounds. Hence, any locally checkable problem whose solution can be checked deterministically in $t(n)=\poly(\log n)$ rounds and admits a randomized algorithm that runs in $r(n)=\poly(\log n)$ rounds also has a deterministic algorithm that runs in $(r(n)+t(n))\cdot \poly(\log n) = \poly(\log n)$ rounds.
\end{proof}
\subsection{Other Implications (Deterministic $\&$ Randomized)}
\label{subsec:OtherImplications}
Here, we mention some of the other implications. This list is not exhaustive; these are just some of the prominent instances that came to our mind. A more thorough job is needed to re-examine all the related literature and list all the consequences.  Moreover, in the interest of brevity and due to the large number of the implications, here we just provide a brief and sometimes informal explanation of each problem; the precise setup can be found in the references that we mention.

\subsubsection{Lov\'{a}sz Local Lemma and the Sublogarithmic Complexity Lanscape}

The Lov\'{a}sz Local Lemma has turned out to have a fundamental role in several distributed problems, and perhaps most remarkably, in the complexity of the locally checkable problems that have sublogarithmic complexity. We next review the LLL problem and outline the new result.

\paragraph{Lov\'{a}sz Local Lemma} Consider a probabilistic setting of events defined on a set of random variables. There is one node for each \emph{bad} event, and $p$ denotes the maximum probability among these bad events. Moreover, each two bad events that share a variable are connected via an edge, and we use $d$ to denote the maximum degree of this graph. The Lov\'{a}sz Local Lemma proves that if $epd \le 1$, then there is an assignment to the variables that avoids all the bad events. In the distributed version of this problem, the question is to efficiently compute such as assignment that avoids all the bad events, where the \local-model graph is the same as the dependency graph among the events. See \cite{chung2017LLL, chang2017time, ghaffari-lll, ghaffari2018derandomizing}. 

\paragraph{Improved deterministic LLL} By running the $O(\log^2 n)$-round randomized distributed LLL algorithm of Moser and Tardos\cite{mt} through the derandomization method of \Cref{thm:derandomization}, we get a $\poly(\log n)$ round deterministic distributed algorithm for Lov\'{a}sz Local Lemma:  

\begin{corollary}
There is a deterministic distributed algorithm that solves the Lov\'{a}sz Local Lemma problem in $\poly(\log n)$ rounds, so long as the maximum probability among the  bad events $p$ and the maximum dependency degree among them $d$ satisfy $epd \leq 1-\delta$, for any constant $\delta>0$ or even a slightly sub-constant $\delta>1/\poly(\log n)$.
\end{corollary}

\paragraph{Improved randomized LLL} By plugging this deterministic Lov\'{a}sz Local Lemma algorithm into the frameworks of \cite{ghaffari-lll, ghaffari2017complexity}, we get a randomized LLL algorithm with complexity $\poly(\log\log n)$ in constant-degree graphs.

\begin{corollary}
There is a randomized distributed algorithm that solves the Lov\'{a}sz Local Lemma problem in $O(d^2)+\poly(\log\log n)$ rounds, so long as the maximum probability among the  bad events $p$ and the maximum dependency degree among them $d$ satisfy $Cpd^8 \leq 1$, for some constant $C>1$.
\end{corollary}

This $\poly(\log\log n)$ round complexity for constant-degree graphs almost settles a conjecture of Chang and Pettie\cite{chang2017time}; their conjecture postulates the existence of an $O(\log\log n)$ time algorithm. 

\paragraph{Complexity of LCLs in the sublogarithmic landscape} 
Due to a beautiful result of Chang and Pettie\cite{chang2017time}, this improved LLL has a remarkable complexity-theoretic consequence:

\begin{corollary}
Any locally-checkable problem that admits an $o(\log n)$ round randomized distributed algorithm in constant-degree graphs also admits a $\poly(\log\log n)$ round randomized algorithm.
\end{corollary}
That is, for any problem whose solution can be checked deterministically in $O(1)$ rounds, in bounded degree graphs, the randomized complexity is either $\Omega(\log n)$ and above, or $\poly(\log\log n)$ and below. As soon as we can prove some LCL problem to admit an $o(\log n)$-round algorithm, we immediately get a $\poly(\log\log n)$ round algorithm.
 
\subsubsection{Packing/Covering Integer Linear Programs}
Covering and packing integer Linear Programs are LPs in the standard form where all the coefficients are non-negative; the former is a minimization problem and the latter is a maximization problem. A wide range of optimization problem can be formulated in this manner. 

A general result of Ghaffari, Kuhn, and Maus~\cite[Section 7]{ghaffari2017complexity} shows that for any covering or packing integer linear program, there is a $\poly(\log n/\eps)$ round randomized algorithm in the \local model for computing a $1+\eps$ (integral) approximation. 
The concrete distributed formulation of these LPs is that we have a bipartite graph where each node on the left shows one of the variables and each node on the right shows one of the constraints, and a constrain node is connected to the variable nodes that it includes. Cf. \cite{ghaffari2017complexity} for details. We note that one can imagine a number of other natural formulations of the optimization problem as a graph, but in the \local model, these usually can simulate each other with a constant round complexity overhead. 

By plugging our network decomposition into the framework of \cite{ghaffari2017complexity}, we can derandomize their result and get a deterministic variant:
\begin{corollary}
For any covering or packing integer linear program, there deterministic algorithm in the \local model that computes a $1+\eps$ approximation in $\poly(\log n/\eps)$ rounds.
\end{corollary}

We note that the conference version of \cite{ghaffari2017complexity} describes the method explicitly only for the maximum independent set problem, but the same technique extends to other covering or packing integer linear programs, as outlined in \cite{ghaffari2017complexity}. A full description will appear in the journal version of \cite{ghaffari2017complexity}. As some concrete examples, this implies $\poly(\log n/\eps)$-round deterministic \local-model algorithms for $1+\eps$ approximation of maximum independent set (as a sample packing problem) and for $1+\eps$ approximation of minimum dominating set (as a sample covering problem). It should be remarked that the \local model does not bound the time for local computation in one node and these two particular results take advantage of that.

\subsubsection{Defective and Frugal Colorings}

\paragraph{Defective coloring} The defective coloring problem is a variant of the standard proper coloring problem, which has turned out to be important in the study of distributed graph algorithms. In an $f$-defective coloring, we allow each node to have up to $f$ neighbors in its own color --- in return for this relaxation, we hope for a smaller number of colors. Open Problem 11.7 in the book of Barenboim and Elkin asks for ``\emph{an efficient distributed algorithm for computing a $O(\Delta/p)$-defective $O(p)$-coloring}". 

We note that an iterative-improvement algorithm of Lov\'{a}sz\cite{lovasz1966decomposition}---which starts with an arbitrary coloring and changes node colors one by one, so long as that improves the node's defect--- ensures the existence of such a defective coloring in all graphs. Kuhn\cite{kuhn2009weak} showed that a $\Delta/p$-defective $O(p^2)$ coloring can be computed in $O(\log^* n)$ rounds. Chung, Pettie, and Su\cite{chung2017LLL} gave a randomized algorithm that in $O(\log n)$ rounds computes an $O(\Delta/p)$-defective $O(p)$ coloring. By running their randomized algorithm through our derandomization result (\Cref{thm:derandomization}), we get an efficient deterministic variant which settles Open Problem 11.7:

\begin{corollary}
There is a deterministic distributed algorithm in the \local model that, for any $p$, computes an $O(\Delta/p)$-defective $O(p)$ coloring in $\poly(\log n)$ rounds. 
\end{corollary}

\paragraph{Frugal coloring} A $k$-frugal coloring is a coloring where each color appears at most $k$ times in the neighborhood of each node (independent of the color of that node itself, which is what makes this definition different from defective coloring). We are not aware of any deterministic distributed algorithm for frugal coloring (with good parameters), but there are some efficient randomized algorithms: Chung, Pettie, and Su\cite{chung2017LLL} show a randomized algorithm that computes an $O(\log^2 \Delta/\log\log \Delta)$-frugal $\Delta+1$ coloring in $O(\log n)$ rounds of the \local model, and a $\beta$-frugal $O(\Delta^{1+1/\beta})$-coloring in $O(\log n\log^2 \Delta)$ rounds of the \local model. By derandomizing these algorithms, we get

\begin{corollary}
There are deterministic distributed algorithm that in $\poly(\log n)$ rounds of the \local model compute 
\begin{itemize}
    \item[(I)] a $O(\log^2 \Delta/\log\log \Delta)$-frugal $\Delta+1$ coloring, and 
    \item[(II)] $\beta$-frugal $O(\Delta^{1+1/\beta})$-coloring.
\end{itemize}
\end{corollary}

\subsubsection{Forest Decomposition and Low Out-degree Orientation}
Consider a graph with arboricity at most $a$, that is, a graph where edges can be decomposed into $a$ forests. Due to a result of Barenboim and Elkin\cite{barenboim10}, there is a deterministic distributed algorithm that decomposes any graph of arboricity $a$ into $2a$ forests, in $O(\log n)$ rounds. In Open Problem 11.10 of their book\cite{barenboimelkin_book}, Barenboim and Elkin ask for an ``\emph{efficient distributed algorithm for computing a decomposition of graph with arboricity a into less than 2a forests}". A result of \cite{ghaffari2017orinetation} provides a randomized $\poly(\log n)$ round algorithm that decomposes the graph into $(1+o(1))a$ forests, when $a=\Omega(\log n)$, and into $(1+o(1))a$ pseudo-forests when $a=o(\log n)$. Recall that a pseudo-forest is an undirected graph where each connected component has at most one cycle. In both cases, the decomposition provides an orientation of the edges where each node has out-degree at most $(1+o(1))a$. To the best of our knowledge, in all distributed applications of the aforementioned forest decomposition, a decomposition into pseudo-forests (or alternatively, just the orientation with the bounded out-degree) would also suffice. Plugging this randomized algorithm into our derandomization result (\cref{thm:derandomization}), we get an algorithm that almost settles Open Problem 11.10 of \cite{barenboimelkin_book}:

\begin{corollary}
There is a deterministic $\poly(\log n)$ round algorithm in the \local model that, for any graph with arboricity at most $a$, computes an orientation with maximum outdegree at most $(1+o(1))a$. Moreover, the algorithm
decomposes the graph into $(1+o(1))a$ forests, if $a=\Omega(\log n)$, and into $(1+o(1))a$ pseudo-forests if $a=o(\log n)$.
\end{corollary}

\subsubsection{Derandomizations in the \congest Model: Neighborhood Cover, Spanners, and Dominating Set} \label{subsub:congestDerand}
 We have already mentioned that our network decomposition algorithm extends to the \congest model, and even has the nice property that each edge is in $\poly(\log n)$ many Steiner trees. We used these to derive our \congest model efficient deterministic MIS algorithm, in \Cref{thm:congest-MIS}. 
 But there is one more generality of our network decomposition, which opens the road for other applications: the algorithm readily extents to powers $G^k$ of the graph $G$, where we connect any two nodes within distance $G$. As stated in \Cref{thm:main_alg_gk}, in $\poly(\log n)$ rounds of the \congest model, we can compute a decomposition into clusters, each with a Steiner tree of depth $\poly(\log n)$, colored with $\poly(\log n)$ colors so that any two clusters wihin distance $k$ have different colors. Moreover, each edge is used in $\poly(\log n)$ Steiner trees. This can be directly plugged into some of the recent work on derandomization in the \congest model, for particular graph problems, to improve the related round complexities. We overview these next.
 \smallskip
 
\paragraph{Sparse neighbohood covers} One prominent corollary is that we get an efficient deterministic algorithm in the \congest model for the \emph{sparse neighborhood cover} problem --- one of the central and versatile algorithmic tools in the study of locality-sensitive distributed graph algorithms\cite{peleg00, awerbuch89}. This corollary follows from using our improved network decomposition in the method provided by Ghaffari and Kuhn~\cite{ghaffari2018congest-derandomizing}. 

\begin{corollary}
There is a deterministic distributed algorithm that  for any radius $r\geq 1$, computes an $\poly(\log n)$-sparse neighborhood cover of the $r$-neighborhoods of the graph, with clusters of radius $r \poly(\log n))$, in $r\poly(\log n)$ rounds of the \congest model. In other words, this gives a clustering of the graph into overlapping clusters of radius $r\poly(\log n)$ such that for each node, its $r$-hop neighborhood is entirely contained in at least one of the clusters and moreover, each node is in at most $\poly(\log n)$ clusters.
\end{corollary}

We note that the above neighborhood cover also settles a question of Elkin~\cite{elkin2006faster}, giving a deterministic variant of his minimum spanning tree algorithm with the same round complexity up to logarithmic factors. \smallskip

\paragraph{Spanner} Another example is the first efficient deterministic distributed algorithm, in the \congest model, for constructing spanners with almost optimal parameters. This follows from plugging our network decomposition into the algorithms of \cite{ghaffari2018congest-derandomizing}:
\begin{corollary}
There is a deterministic distributed algorithm that in $\poly(\log n)$ rounds of the \congest model, computes a spanner with stretch $2k-1$ and size $O(k n^{1+1/k} \log n)$.
\end{corollary}

\paragraph{Dominating set and set cover} As another example, by putting together our \congest-model network decomposition with the work of Deurer et al.\cite{Deurer2019}, we get the first efficient deterministic \congest model approximation of minimum dominating set. Moreover, as outlined in \cite{Deurer2019}, this can also be extended to an approximation of set cover. These lead to the following corollary:
\begin{corollary}
There are $\poly(\log n)$-round deterministic distributed algorithms in the \congest model that compute: (I) a $(1+o(1))\log \Delta$ approximation of minimum dominating set, where $\Delta$ denotes the maximum degree, and (II) a $(1+o(1))\log \Delta$ approximation of the minimum set cover problem, where $\Delta$ denotes the maximum set size.
\end{corollary}
\smallskip

%\newpage
\section*{Acknowledgment}
We are grateful to Christoph Grunau for several discussions about verifying the ideas and working intensively with us throughout the writing process. We also thank Sebastian Brandt, Keren Censor-Hillel, Yi-Jun Chang, Davin Choo, Michael Elkin, Fabian Kuhn, Merav Parter, Julian Portmann, and Hsin-Hao Su for proofreading an earlier version of this write-up and helpful comments. We are also grateful to the reviewers of STOC 2020 for their helpful comments.

The first author also thanks Michael Elkin and Jukka Suomela for very inspiring discussions about network decomposition.

%for an inspiring discussion about the strength of the assumption of polynomially bounded identifiers of each node (in a computational model similar to \local).

\bibliography{ref}

\fullOnly{\bibliographystyle{alpha}}
\shortOnly{\bibliographystyle{ACM-Reference-Format}}

\appendix
\fullOnly{ %---------------------------------

\section{Comparison of previous work with \Cref{thm:decomp-informal}}
\label{sec:appendix_relation}

Recall that before our work, the state of the art deterministic algorithm for network decomposition was the \ps-round algorithm of Panconesi and Srinivasan\cite{panconesi-srinivasan}. This result itself was a refinement of the approach of Awerbuch et al.\cite{awerbuch89}, who obtained an \aglp-round algorithm for network decomposition. 
In this section, we first briefly recall this approach, in \Cref{subsec:RecapAwerbuch}. Then, in \Cref{subsec:newSlow}, we describe a new algorithm that also achieves this \aglp round complexity. In some sense, the approaches of \cite{panconesi-srinivasan, awerbuch89} and the algorithm that we describe here can be viewed as two fundamentally different methods, both of which seem to get stuck at the \aglp round complexity. Finally, in \Cref{subsec:fastVSslow} then discuss how these compare with the approach of \Cref{thm:main} that achieves a $poly(\log n)$ round complexity.

\subsection{Recap on Ruling sets}
Both the approach of Awerbuch et al. \cite{awerbuch89} and of the new algorithm that we present here would achieve round-complexity \ps, assuming we have an efficient algorithm for maximal independent set (MIS). However, the only known way of computing MIS is via network decomposition (cf. \cref{thm:MIS-local}). 
Hence, instead, we will use a \emph{ruling set} procedure. As we recall next, ruling set is a certain weakening of MIS, and it can be computed deterministically in $O(\log n)$ rounds\cite{awerbuch89} (for certain parameters). 
%Next, we define the notion of $(\alpha, \beta)$-ruling set. 

%While the beautiful algorithms of Luby \cite{luby86} and Alon, Babai, and Itai \cite{alon86} show that a maximal independent set can be obtained in polylogarithmic round complexity in randomized setting, the question of a deterministic algorithm was open prior to our work. 
%However, we can deterministically construct a certain weakening of maximal independent set that is a special case of a so-called ruling set. 

\begin{definition}
Given a graph $G$, an $(\alpha, \beta)$-ruling set is a subset $S \subseteq G$ such that
\begin{itemize}
    \item for each two different vertices $u, v \in S$ their distance in $G$ is at least $\alpha$,
    \item for each vertex $u \in G$ there is a vertex $v \in S$ such that the distance of $u$ and $v$ in $G$ is at most $\beta$. 
\end{itemize}
\end{definition}

To give an example, a maximal independent set is simply a $(2,1)$-ruling set. 
For the sake of completeness, we now show how to construct a $(2, O(\log n))$-ruling set, as presented by Awerbuch et al.\cite{awerbuch89}. 
In other words we find a set $S$ such that no two of its vertices are neighbouring and for each vertex $v$ there is a vertex $u\in S$ within its $O(\log n)$ distance. 
%The second condition is weakened with respect to a maximal independent set, where we require that for each 

\begin{proposition}
\label[proposition]{prop:ruling_set}
There is a deterministic $O(\log n)$-round complexity algorithm that computes a $(2, O(\log n))$-ruling set of any graph $G = (V,E)$ in the \local model, given that each vertex holds a unique $O(\log n)$-bit identifier. 
\end{proposition}

\begin{proof}
We start with a set $S_0 = V$ of all vertices and gradually prune this set in $O(\log n)$ steps, producing a chain $V = S_0 \supseteq S_1 \supseteq \dots \supseteq S_{O(\log(n))}=:S$, where the final set $S$ in the chain is the desired $(2, O(\log n))$-ruling set. 
%This gradual approach is motivated by the need of ensuring that there are no two neighbouring vertices in the final set $S$ -- instead of tackling this challenge directly, we will only ensure in the $i$-th round that $S_i$ does not contain a pair of neighboring vertices that would have the same bit at the $i$-th position of their identifiers. 

This is done as follows: in the $i$-th step, each vertex $u \in S_{i-1}$ that has $1$ at the $i$-th position of its identifier checks whether it has a neighbour with identifier containing $0$ at its $i$-th bit. 
If this is the case, $u$ decides to leave $S_{i-1}$. 

This process clearly finishes with a set $S_{O(\log n)}$ without any neighbouring vertices. 
To see that for each vertex $v \in G$ we can find $u \in S$ at distance $O(\log n)$, consider the previous process again and for each vertex $w$ that is removed from the ruling set at some iteration $i$, draw an oriented edge towards its neighbour that made $w$ disappear from $S_{i-1}$. 
These edges form disjoint oriented trees of depth $O(\log n)$ and with vertices of $S$ being their roots. This implies that $S$ is indeed $(2, O(\log n))$-ruling set. 
\end{proof}

\subsection{Recap on the Approach of Awerbuch et al. \cite{awerbuch89}}
\label{subsec:RecapAwerbuch}
Next, we sketch the \aglp-round algorithm for network decomposition from Awerbuch et al. \cite{awerbuch89}. We note that the algorithm of Panconesi and Srinivasan\cite{panconesi-srinivasan} is a improvement and refinement of this approach; while the algorithms is different, the overall approach outline is the same. In particular, both of these algorithms seem to fall short of reaching a $\poly(\log n)$ complexity, for the same reasons.

\begin{proposition}
There is a deterministic \aglp-round algorithm in the \local model that computes a network decomposition of the underlying graph $G$ into $O(\log n)$ color classes, each one with clusters of diameter at most $O(\log n)$. 
\end{proposition}

\begin{proof}[Proof Sketch]
It suffices to construct a decomposition with \aglp color classes and clusters of weak-diameter \aglp \footnote{I.e., any two vertices of a cluster have $2^{O\left(\sqrt{\log n \log \log n}\right)}$ distance in $G$. }, since then we can reduce the diameter and the number of color classes by known reduction~\cite{awerbuch96}. See also \Cref{thm:main-strong} for such a transformation. 

We show an iterative process with a parameter $d$ that will have $\log_d n$ iterations. 
The $i$-th iteration has round complexity $O(\log n)^{i-1} \cdot O(d^2 + \log^* n)$. 
This process will construct network decomposition with $d \log_d n$ colors and with the diameter of each cluster being $O(\log n)^{\log_{d} n}$. 
We optimize $d$ by setting $d = 2^{O\left(\sqrt{\log n \log \log n}\right)}$ which gives the advertised parameters. 

Now we describe one phase of the process that works with a graph $G_i = (V_i, E_i)$. Initially, we set $G_0 := G$. We shall ensure that one round of communication in $G_i$ can be simulated in $O(\log n)^{i-1}$ rounds of communication in $G$. 

We split the vertex set $V_i$ into the set $H_i$ of vertices of degree at least $d$ and $L_i$ of vertices of degree smaller than $d$. 
To deal with low degree vertices, we construct their $d^2$-coloring by the algorithm of Linial\cite{linial1987LOCAL} in time $O(\log n)^{i-1} \cdot O(d^2 + \log^* n)$ -- here the first term is the number of rounds in $G$ we need to simulate one round in $G_{i-1}$ and the second term is the complexity of Linial's algoirthm. 
Each constructed color class will constitute one class of the final network decomposition. 

Next, the nodes locally construct a graph $G_i^2$ where two vertices are connected if their distance in $G_i$ is at most $2$. % (we also include already colored nodes of $L_i$). 
Now we use \Cref{prop:ruling_set} to find a $(2, O(\log n))$ ruling set in $G_i^2$ restricted to vertices of $H_i$. 
The ruling set procedure also gives us a decomposition of $H_i$ into oriented trees of depth $O(\log n)$, hence we get a decomposition of vertices of $H_i$ into clusters of diameter $O(\log n)$ in $G_i$. 
These clusters will be vertices of $G_{i+1}$ and there is an edge in $E(G_{i+1})$ for each pair of clusters adjacent in $G_i$. 

In the next round, the additional communication overhead of the algorithm on $G^{i+1}$ is only $O(\log n)$ times higher than for $G^i$. 
More precisely, each vertex of $G^{i+1}$ corresponds to a set of vertices of $G$ that have weak-diameter at most $O(\log n)^{i+1}$. 
To finish the argument we need to bound the number of iterations. This follows from the fact that each vertex in the ruling set found in the $i$-th iteration has degree at least $d$ in $G_i$. 
Moreover, the set of its neighbours is disjoint with the neighbours of all other vertices in the ruling set due to its construction in graph $G_i^2$. 
Hence, the number of vertices in $G_{i+1}$ is at least $(d+1)$ times smaller than the number of vertices of $G_{i}.$ Thus the number of iterations is bounded by $\log_{d} n$. 
\end{proof}

\subsection{A New Network Decomposition with Complexity \aglp} 
\label{subsec:newSlow}
We now propose a new algorithm, which appears to be fundamentally different and it also achieves the same \aglp round complexity. In the next subsection, we will contrast these two approaches with the $\poly(\log n)$-time algorithm described in \Cref{thm:main}. 

Similarly to the case of \Cref{thm:decomp-informal}, we only show how to find first $O(\log n)$ color classes of the resulting decomposition that ensure that the overall number of uncolored vertices drop by a factor of $(1 - 2^{-O(\sqrt{\log n\log\log n}})$. The result then follows after repeated application of \Cref{prop:rapid_ball_growing} and, finally, after reduction of the number of colour classes and their diameter using the transformations of \cite{awerbuch96} (See also \Cref{thm:main-strong} for such a transformation). 

The idea of the new algorithm is to simulate the sequential process of building network decomposition from the proof of \Cref{thm:main-strong}. 
In particular, each vertex $u$ will consider a ball $B(u, r(u))$ around itself of radius $r(u)$. This is similar to the sequential ball carving process in the proof of \cref{thm:main-strong}. 
The difference is that, here, instead of additively growing the radius of the ball by one in each step, we will consider radii that are powers of a small value $t$ (later fixed to be $t=O(\log n)$). We discuss this choice after the proof of \cref{prop:rapid_ball_growing}. 
%, where $r(u)$ is chosen by a doubling process \mtodo{what is a doubling process?}, similarly to the sequential algorithm from the proof of \cref{thm:main-strong}. 
% We shall call such a ball $B(u, r(u))$. 
%and the aim is to achieve that i relatively small radius $r(u)$ and small boundary (we make this precise later). 
The balls of different vertices will possibly intersect considerably. The idea is then to output only a subfamily $I := \{B(u_1, r(u_1)), B(u_2, r(u_2), \dots, B(u_k, r(u_k)) \}$ composed of balls that do not intersect nor touch, pairwise. 
In the actual proof we follow this idea, though the resulting output is little bit more complicated.  
%(it will be little bit more complicated due to the fact that different vertices may grow balls of different radius -- we will actually output $O(\log n)$ such families). 

The main challenge here is to ensure that the family $I$ covers a substantial part of the whole graph. 
To this end, we need a stronger condition, which is effectively like thinking of a ``wider boundary" for each ball. Concretely, instead of requiring a lower bound on the ratio $|B(u, r(u))| / |B(u, r(u)+1)|$--- as we had in the sequential algorithm of \cref{thm:main-strong} and our main distributed algorithm from \cref{thm:decomp-informal} --- we require a lower bound on the ratio $|B(u, r(u))| / |B(u, O(\log n) \cdot r(u)|$. This allows us to find such a good set $I$ with the help of ruling sets. 
The above condition leads to an optimization problem giving the same round complexity as in \cite{awerbuch89}.

\begin{proposition}
\label[proposition]{prop:rapid_ball_growing}
There is a deterministic \local-model algorithm that, for any graph $G=(V,E)$ with $n$ nodes, in \aglp rounds, computes a $S \subseteq V$ of vertices and a coloring of this set with $O(\log n)$ colors such that (1) each each connected component of any color class of $S$ has weak-diameter \aglp, and (2) $|S| \ge |V| / $\aglp. 
\end{proposition}

\begin{proof}
First fix $t = O(\log n)$ to be such that \Cref{prop:ruling_set} can give us $(2,t/4)$-ruling sets in $O(\log n)$ rounds. 
Moreover, set $\eps = 2^{-\sqrt{\log n / \log\log n}}$. 
Each vertex $u$ gradually grows a sequence of balls $B(u, t^i)$ for $i \ge 0$. 
It stops growing the ball around it at the first point when $|B(u, t^i)| \ge \eps \cdot |B(u, t^{i+1})|$ and then sets $r(u) := t^{i}$. 
Note that the maximum $i$ the vertex needs to consider is bounded by 
\[
\log_{1/\eps} n 
= \frac{\log n}{\log 1/\eps}
= \sqrt{\log n \log\log n}.
\]
This is because, in each step, the vertex either finishes growing or the volume of its ball grows by multiplicative factor of at least $1/\eps$. 
Therefore, the overall radius the vertex needs to consider and, thus also the round complexity of this step, is bounded by 
\[
t^{\log(1/\eps)} = 2^{O(\log\log n) \cdot \sqrt{\log n / \log\log n}} = 2^{O(\sqrt{\log n \log\log n})}. 
\]

Next, we define $\log_{1/\eps} n+1$ graphs $G_0, G_1, \dots, G_{\log_{1/\eps} n}$, where each vertex $u$ is in the vertex set of $G_i$ if and only if $r(u) = t^i$. 
Two vertices $u,v \in G_i$ are connected in $G_i$ if their distance in $G$ is at most $3 \cdot t^i$. This is to ensure that if there is not an edge between $u$ and $v$ in $G_i$, the corresponding balls $B(u,r(u))$ and $B(v,r(v))$ neither overlap, nor are adjacent. 

Now for each $G_i$ we run the $(2,t/4)$-ruling set algorithm of \Cref{prop:ruling_set} on $G_i$ in $O(\log n) \cdot 2^{O(\sqrt{\log n \log\log n})} = 2^{O(\sqrt{\log n \log\log n})}$ rounds. 
We call the output ruling set $R_i$ and define the set $S$ of vertices we color as
\[
S = \bigcup_{0 \le i \le \log_{1/\eps}n}\;\; \bigcup_{u \in R_i} B(u, r(u)). 
\]
Moreover, each vertex $v \in S$ is colored by the smallest index $i$ such that $v$ is in some ball $B(u,r(u))$ for a vertex $u \in R_i$. 
%defines a family of sets $\mathcal{I}_i$ as follows. For fixed $i$ and each $u \in G_i$ chosen in the ruling set of $G_i$ we add the ball $B(u, r(u))$ to the family $\mathcal{I}_i$. 
%The output $S$ of the algorithm consists of vertices of $\bigcup \mathcal{I}_0 \,\cup\, \bigcup \mathcal{I}_1 \,\cup\, \dots \,\cup\, \bigcup \mathcal{I}_{\log_{1/\eps} n}$, where each vertex $u \in S$ is colored by the smallest index $i$ such that $u$ is in some $B(v,r(v))$ of a $v \in \mathcal{I}_i$. 

Next, we argue that this output satisfies all the desired conditions. First, notice that the number of color classes we output is $\log_{1/\eps} n = \sqrt{\log n \log\log n} = O(\log n)$. 
For each color class we output a subset of balls with diameter $2^{O(\sqrt{\log n \log\log n})}$ that were neither overlapping, nor adjacent due to the construction of the ruling set and hence their subsets have still bounded weak-diameter and they are neither overlapping, nor adjacent. 

It remains to be seen that the number of vertices of $S$ constitutes a substantial fraction of the vertices of $G$. 
The crucial observation is that due to the property of our ruling sets, each $v \in G_i$ has distance in $G$ at most $t/4 \cdot 3t^i < t^{i+1}$ to some $u \in R_i$. Hence, while $S = \bigcup_{0 \le i \le \log_{1/\eps}n} \bigcup_{u\in R_i} B(u, r(u))$ is a union of $O(\log n)$ disjoint sets, the union $\bigcup_{0 \le i \le \log_{1/\eps}n} \bigcup_{u\in R_i} B(u, t \cdot r(u))$ already covers the whole $G$. 
Finally, recall that each ball $B(u, r(u))$ already substitutes a substantial fraction of the bigger ball $B(u, t \cdot r(u))$. This enables us to finish off with the following simple double counting argument. 

\begin{align*}
    |S| &=
    \left\lvert  \bigcup_{0 \le i \le \log_{1/\eps}n}\;\; \bigcup_{u \in R_i} B(u, r(u))\right\rvert \\ 
    &\ge \frac{1}{\log_{1/\eps}n} \sum_{0 \le i \le \log_{1/\eps}n} \left\lvert \bigcup_{u \in R_i} B(u, r(u))\right\rvert 
    && \text{each vertex is overcounted at most $\log_{1/\eps}n$ times}\\ 
    &= \frac{1}{\log_{1/\eps}n} \sum_{0 \le i \le \log_{1/\eps}n} \sum_{u\in R_i} \left\lvert B(u, r(u))\right\rvert 
    && \text{balls with centers from the same ruling set are disjoint}\\
    &\ge \frac{\eps}{\log_{1/\eps}n} \sum_{0 \le i \le \log_{1/\eps}n} \sum_{u\in R_i} \left\lvert B(u, t \cdot r(u))\right\rvert 
    && \text{$|B(u,r(u))| \ge \eps |B(u, t\cdot r(u))|$} \\
    &\ge \frac{\eps}{\log_{1/\eps}n} \left\lvert  \bigcup_{0 \le i \le \log_{1/\eps}n} \bigcup_{u\in R_i} B(u, t \cdot r(u))\right\rvert \\
    &\ge \frac{\eps}{\log_{1/\eps}n} \cdot n && \text{this union covers the whole graph}\\
    &= n / 2^{O(\sqrt{\log n \log\log n})}. 
\end{align*}
\end{proof}

Note that choosing the radius as a power of $t$, instead of additive increases similar to the sequential ball carving process, is crucial in the above proof. This is because it is necessary to ensure that not only the boundary of each ball $B(u,r(u))$ is small as in the sequential ball carving process, but also the volume of the ball $B(u, t \cdot r(u))$ is not much bigger than the volume of $B(u, r(u))$, as the factor $|B(u,r(u))| / |B(u, t \cdot r(u))|$ is also the fraction of all vertices that we output in the end output as the set $S$.

\subsection{Contrasting the $\poly(\log n)$-round and \aglp-round algorithms}
In this subsection, we discuss the similarities in the above two approaches and describe why they both fall short of reaching a $\poly(\log n)$ round complexity. In particular, we discuss one particular graph on which both of these approaches get stuck at a $2^{O(\sqrt{\log n})}$ complexity. We note that the discussions in this subsection are not written formally, and they instead try to provide an intuitive explanation of the shortcoming of the two approaches, which is circumvented by the $\poly(\log n)$-round algorithm of \Cref{thm:main}. 

\label{subsec:fastVSslow}
First, let us note the similarities between the algorithm from \cite{awerbuch89} and \Cref{prop:rapid_ball_growing}. 
Both algorithms use a ruling set computation just as a replacement for a maximal independent set, which allows efficient computation. 
It is straightforward to observe that, if we could solve the maximal independent set problem in $\poly( \log n)$ rounds, then the complexity of both of these approaches would improve from \aglp to $2^{O(\sqrt{\log n})}$. This is roughly because, in that case, per iteration, we would lose an $O(1)$ factor instead of an $O(\log n) = O(2^{\log\log n})$ factor.
For the approach of \cite{awerbuch89}, the algorithm of \cite{panconesi-srinivasan} cleverly circumvents this issue via additional insights (which recursively build the desired maximal independent set, at each point, using the network decomposition that is constructed up to that point). 
% (here, oracle access means ``assume that we can compute a maximal independent set, say, in $\poly(\log n)$ rounds"). 

Moreover, both algorithms cannot bound the increase of the radii of growing clusters in an additive manner. Instead, they both the increase in the radii by a multiplicative factor of $O(\log n)$ (or $O(1)$ should they have access to an efficient algorithm for maximal independent set). 
They have to achieve as much as possible for one multiplicative increase and here their strategies differ: 
The approach of Awerbuch et al.\cite{awerbuch89} uses iterations of \textit{contracting clusters}, while the algorithm of \Cref{prop:rapid_ball_growing} uses a sped up \textit{ball carving}. More concretely, the algorithm of \cite{awerbuch89} decreases the overall number of vertices by a factor of roughly $2^{O(\sqrt{\log n})}$, via contractions, the algorithm from \Cref{prop:rapid_ball_growing} multiplies volume of a given ball by the same factor, via growing the ball's radius. 

Despite this departure, it appears that both approaches fall short of reaching a $\poly(\log n)$ round complexity for similar reasons. In particular, next we discuss an example graph $G$ that is ``hard" for both of these approaches; that is, both approaches achieves at best a $2^{O(\sqrt{\log n})}$ round complexity on this particular graph. This is a simple high-dimensional torus-like graph $G$, as follows: The vertices of $G$ are vectors of length $\sqrt{\log n}$ with coordinates from the set $\{ 0, 1, \dots, 2^{\sqrt{\log n}} -1 \}$. Two vertices/vectors are connected via an edge if and only if the difference in each coordinate is at most one modulo $2^{\sqrt{\log n}}$. We will refer to the first parameter, i.e., the length of vectors which is set to $\sqrt{\log n}$, as the \textit{dimension}, and to the second parameter, i.e., the maximum coordinate value which is set to $2^{\sqrt{\log n}}$, as the \textit{side length}.

The vertex-transitive graph $G$ has the following two properties. First, its diameter is $\Theta(2^{\sqrt{\log n}})$. Second, the graph has a ``volume expansion" factor of $\Theta(2^{\sqrt{\log n}})$ as we double the radius. That is, for each vertex $u \in V(G)$ and for any radius $r \leq \Theta(2^{\sqrt{\log n}})$, we have $|B(u, 2r)| / |B(u, r)| = (4r+1)^{\sqrt{\log n}} / (2r+1)^{\sqrt{\log n}} = \Theta(2^{\sqrt{\log n}})$. In fact, $G$ optimizes the trade-off between these two parameters, diameter and ``volume expansion". 

The two \aglp-round algorithms in fact optimize the same trade-off and we now observe that the approaches used there cannot yield $2^{o(\sqrt{\lg n})}$ running time. 
Let us first observe it on the example of algorithm from \Cref{prop:rapid_ball_growing}. 
There, we are growing a ball around each vertex and doubling its diameter at every step. 
If we continue doubling the ball until it covers the whole graph, its radius reaches the diameter of $G$ and we will need $O(2^{\sqrt{\log(n)}})$ rounds. 
If we, instead, stop growing the ball with any radius $r$, we will have $|B(u, O(\log n)\cdot r)| \ge |B(u, 2r)| = \Theta(2^{\sqrt{\log n}}) \cdot |B(u,r)|$. Hence, we can guarantee that only $\Theta(1 / 2^{\sqrt{\log n}})$ fraction of vertices will be colored in one phase of the algorithm. 

Similarly, if we run the algorithm of Awerbuch et al.~\cite{awerbuch89}, we will either decide that all vertices of $G$ have small degree and color them in $2^{\Theta(\sqrt{\log(n)})}$ rounds, or we decide that they have high degree. In the latter case, we compute a ruling set and use it to contract vertices and form a new cluster graph. However, one can see that for an arbitrary ruling set --- e.g., if the ruling set is given by vertices with all coordinates equal $0$ modulo $3$ and vertices outside ruling sets join the cluster of the unique adjacent vertex in the ruling set --- %(that is, unless we assume additional nice properties from this ruling set), 
the new cluster graph will again be a $\sqrt{\log(n)}$-dimensional torus, this time with coordinates from the set (up to rounding errors) $\{0, 1, \dots, 2^{\sqrt{\log(n)}}/3\}$. That is, the new contracted graph will also be a torus-like graph with the same structure described above and only with a constant factor smaller side length. 
Even after we contract vertices for $O(\sqrt{\log(n)})$ repetitions, we will still have a torus with dimension $\sqrt{\log(n)}$ and side length $2^{\Theta(\sqrt{\log(n)})}$. 
But already at this point, simulating one round of communication on this contracted graph takes $2^{\Theta(\sqrt{\log(n)})}$ communication rounds in the original graph. Hence, the algorithm cannot achieve a round complexity $2^{o(\sqrt{\log(n)})}$. 

%is a hard example for both algorithms. 
%\mtodo{The next sentences is too convoluted. Break it into a few sentences}
%The reason is that, in a certain sense, the right trade-off between expansion  -- growing a ball from a particular starting vertex, its volume increases by a factor of $2^{\Theta(\sqrt{\log n})}$ after we double the radius -- which is its ``local" characteristics, and diameter -- it is again $2^{\Theta(\sqrt{\log n})}$ --  which is its ``global" characteristics. 

Our algorithm from \cref{thm:main} circumvents this issue by growing the clusters more carefully: crucially, we get only additive increase in diameter of each cluster per step, instead of a multiplicative increase as in the examples above. 

%The first author actually constructed the algorithm above while trying to understand whether there is a fundamental reason, why we cannot push the $2^{O(\sqrt{\log n})}$-round complexity to $\poly (\log n)$ for deterministic algorithms. 

%This even lead the first author to believe that there can be a fundamental reason, why the $2^{O(\sqrt{\log n})}$ barrier cannot be pushed to $\poly (\log n)$ for deterministic algorithms. 
%Though the torus-like graph did not yield any insights in this direction, it served as a natural example suggesting that if there is a deterministic algorithm for network decomposition, it should better use in a strong way the option of breaking symmetry via the unique labels of each vertex. 
%This in turn suggests to ``open up" the black box of the ruling set algorithm, which then very naturally leads to \Cref{thm:decomp-informal}. 
}

\fullOnly{
\section{CONGEST Network Decomposition for Power Graphs}
\label{app:congest-powergraph-decomp}
Here, we present the formal proof of \Cref{thm:main_alg_gk}. 
This formal proof expands on the proof sketch provided in \Cref{subsec:CONGEST-construction} and provides addition low-level details. 

\begin{proof}[Proof of \Cref{thm:main_alg_gk}]
We adopt the algorithm from \Cref{lem:main} and the notation used throughout the proof; we also apply the lemma as in the proof of \Cref{thm:main}. The only change is that the process we run in a given step of a given phase will involve all red nodes at distance $k$ from some blue node, instead of only red nodes neighbouring to blue nodes in the original algorithm.
More concretely, all the red nodes in $k$-hop distance of some blue node propose to some blue cluster. This is done as follows. 

We describe a process with $k$ \emph{iterations} that we run in a given step of a given phase. The process can be thought of as a variant of a breadth first search (BFS) algorithm run from all blue nodes at once. 

In the first iteration, each blue node starts with a token with the label  $Y$ of its cluster $S'(Y)$ (we dropped the index $i$ from the original notation $S'_i(Y)$ since we already fixed a phase) and it sends this token to all of its neighbours. 

In the iterations $2$ to $k$, the following happens. 
If at any point of the algorithm any node $u$ from $G$ (here we consider even dead nodes that generally include also nodes of the host graph that were already colored) receives \emph{for the first time} some nonzero number $t$ of tokens, say a set $\{ Y_1, Y_2, \dots, Y_t \}$ with the label $Y_1$ being the smallest, it does the following. 

\begin{itemize}
    \item If $u$ is a living blue node, it does not do anything. 
    \item If $u$ is a living red node, it adds itself as a new terminal node to the underlying Steiner tree of the cluster $S'(Y_1)$ together with an oriented edge pointing towards some node $v$ that sent a token with $Y_1$. %This keeps Invariant A satisfied. 
    The node then sends a token $Y_1$ to all of its neighbours. 
    
    The node $u$ will later propose to the blue cluster with the identifier $Y_1$. 
    To propose, the node will broadcast via the Steiner tree of the cluster $S'(Y_1)$ that it just joined. 

%Since $u$ has distance at most $i$ from some blue node of cluster $Y_1$, all of its neighbours are of distance at most $i + 1$ from that node and, hence, such token does not contradict Invariant A. 
    
    \item 
    If $u$ is a dead node, after receiving tokens in the iteration $i$, $u$ first checks whether in the previous run it received tokens also in the iteration $i$, or later. If the latter is the case, it adds itself as a new nonterminal node to the underlying Steiner tree of the cluster $S'(Y_1)$ together with an oriented edge pointing towards some node $v$ that sent a token with $Y_1$. 
    The node then sends a token $Y_1$ to all of its neighbours. 
    
    However, if, during the last BFS, $u$ received some tokens during the same iteration $i$ and chose to forward a token $Y_\textrm{last}$, it will forward the same token also this time, not considering the tokens it actually received. 

    We will see that it is not possible that $u$ received tokens earlier than in the $i$-th iteration in the previous run of BFS via standard argument about correctness of BFS. 
        
%$i < i_{\textrm{last}}$. 
    %If it is so, it sets $(i_{\textrm{last}}, Y_{\textrm{last}}) \leftarrow (i, Y_1)$. 
    %Then, if it is not already in the Steiner tree of cluster $Y_1$, it adds itself to the tree as a nonterminal node, together with an edge towards $v_1$. %This ensures that both invariants A and B are still satisfied. 
    %Finally, the node sends a token $Y_1$ to all of its neighbours. % which does not contradict Invariant A. 
    
    %If $i = i_\last$, we know from Invariant B that in some previous step this dead node had to receive a token $Y_\last$ in iteration $i$ and sent it to its neighbours in iteration $i+1$. 
    
    %, hence there is a blue node from cluster $Y_\textrm{last}$ of distance at most $k - i_\textrm{last}$ and $u$ is already a nonterminal vertex of a Steiner tree for the cluster $Y_\last$. 
    
    %Hence, the node $u$ is already a nonterminal node of a Steiner tree of the cluster $S'(Y_\last)$. 
    %The node discards all received tokens and sends a token $Y_\last$ to all of its neighbours. 
\end{itemize}

If the node already received some tokens during this breadth first search algorithm, it does not receive or send any more tokens. 

After the breadth first search algorithm finishes, roots of all Steiner trees collect the number of proposing red nodes and each root decides to either accept all proposing red vertices and recolor them to blue, or it makes them die and stops growing with the same decision rule as in \Cref{lem:main}. 
The Steiner trees, however, stay the same even if some of its vertices die, the red nodes that died are just labeled as nonterminals. 
This finishes the description of one step of current phase. 

Besides these changes, the algorithm (the number of phases, steps in each phase and decisions of blue clusters whether to grow or not) stays the same. 

\paragraph{Analysis}
To argue about the correctness of the new version of the algorithm, we first check that some basic properties of BFS. 

\begin{observation}
\label[observation]{obs:bfs}
Throughout the course of the algorithm, the following holds:
\begin{enumerate}
    \item The underlying Steiner trees of clusters are indeed trees oriented towards the root of the cluster throughout the course of the algorithm. 
    \item In the $i+1$-th iteration tokens are sent exactly by vertices whose distance from the set of all blue vertices is $i$.
    \item If a node $u$ sends a token $Y$ in some iteration, it belongs to the Steiner tree of $S'(Y)$ as either terminal or nonterminal node. 
    \item If in some step $j$ a dead node $u$ added itself to the Steiner tree of cluster $S'(Y)$ and later in step $\tilde{j} > j$ it added itself to the Steiner tree of cluster $S'(\tilde{Y})$, the distance of $u$ to the set of blue vertices was strictly smaller in step $\tilde{j}$ than in step $j$. 
\end{enumerate}
\end{observation}
\begin{proof}
The first two bullet points follow from the standard analysis of BFS. 
For the third bullet point, we recall that if a dead node $u$ sends a token $Y_\textrm{last}$ in some iteration, then node $u$ is already part of the Steiner tree of $S'(Y_\textrm{last})$ by induction. 

The last bullet point follows from the fact that a dead node adds itself to a Steiner tree only if it receives token in an earlier iteration than during the previous step, which by the second bullet point means that the distance of $u$ to the set of blue clusters decreased. 

%we consistently break symmetry to smaller identifiers $Y$: if in step $j$ the dead node $u$ received a message from some blue node after $i$ steps, we know that next time it will receive the message from the same blue node after $i$ steps, unless there is now a blue node that is even closer to $u$. 
\end{proof}

%\begin{observation}
%If a node is reached with a token $Y$ in the $i$-th step, its distance from some node from cluster $Y$ is at most $i$. 
%\end{observation}
%\begin{proof}
%This also follows from a standard breadth first search argument. 
%Moreover, we need to observe that if a node decides to send a token $Y_\last$ in the $i+1$-th step, it means that it was already reached by a token $Y_\last$ in the $i$-th step, hence it is of distance at most $i$ to the cluster $Y$, hence the distance of its neighbours is at most $Y + 1$.  
%\end{proof}

Now we can mostly replicate the proof of \Cref{lem:main}. We state equivalents of observations from the proof of \Cref{lem:main} and argue that they are satisfied if they differ substantially from the arguments for \Cref{lem:main}. 
To argue about $k$-hop separation of the clusters, instead of the invariants (I) and (II) from \Cref{lem:main} we keep stronger invariants (I') and (II'). We keep invariant (III) the same. 

\begin{enumerate}
    \item[(I')] For each $i$-bit string $Y$, the set $S'_i(Y)\subseteq S'_i$ of all living nodes whose label ends in suffix $Y$ has no other living nodes $S'_i\setminus S'_i(Y)$ in its $k$-hop neighbourhood. 
    
    In other words, the set $S'_i(Y)$ is a union of some connected components of the subgraph $G[S'_i]$ induced by living nodes $S'_i$ and in the $k$-hop neighbourhood in $G$ around $S'_i(Y)$ all nodes are either dead or they do not belong to the set $S$ (they were colored by previous application of the algorithm). 
    \item[(II')] For each label $L$ and the corresponding cluster $S'_i(L)$, the related Steiner tree $T_L$ has radius at most $i \cdot k \cdot R$, where $R=O(\log^2 n)$. 
    \item[(III)] We have $|S'_{i+1}|\geq |S'_i|(1-1/2b)$.
\end{enumerate}

Now we repeat the list of observations from the analysis of \Cref{lem:main} and remark on them whenever they differ from their counterparts.

\begin{observation}\label[observation]{obs:khop_finite} Any blue cluster stops after at most $4b\log n$ steps. 
\end{observation}
\begin{proof}
The proof stays the same as in \Cref{obs:finite}. 
\end{proof}

\begin{observation}\label[observation]{obs:khop_stop} Once a blue cluster $A$ stops growing, there is no red node in its $k$-hop neighbourhood in $G$ and there never will be one in this phase. 
\end{observation}
\begin{proof} 
As in \Cref{lem:main}, consider the step in which cluster $A$ stops. In that step, each red node in its $k$-hop neighbourhood in $G$ (if there is one) either requested to join $A$ or some other blue cluster. 
This is because even if a dead node $u$ did not forward the token with the identifier of $A$, it forwarded a token of a blue cluster $A'$ that is of the same distance to $u$ or even closer. 

Any red node $v$ that requested to join some blue cluster either adopts a blue label or dies. In either case, $v$ is not a living red node anymore (and it will never become one). From this point onward, this blue cluster $A$ never grows or shrinks. 
\end{proof}

\begin{observation}\label[observation]{obs:khop_inv_two} In each step, the radius of the Steiner tree of each blue cluster grows by at most $k$, while the radius of the Steiner tree of each red cluster does not grow. This implies invariant (II').
\end{observation}

\begin{observation}\label[observation]{obs:khop_inv_three} The total number of red vertices in $S'_i(Y)$ that die during this phase is at most $|S'_i(Y)|/(2b)$. This implies invariant (III).
\end{observation}
\begin{proof} The proof stays the same as in \Cref{obs:inv_three}. 
\end{proof}

Now we can bound the number of Steiner trees that use each particular edge. 

\begin{observation}\label[observation]{obs:khop_congestion} In construction of one color class of the decomposition, each edge is used in $O(\log n \cdot \min(k + \log^2 n))$ Steiner trees. Thus, overall, each edge is used in $O(\log^2 n \min(k + \log^2 n))$ Steiner trees. 
\end{observation}
\begin{proof}
We run through $O(\log n)$ phases and in each phase we run $O(\log^2 n)$ steps. 
First, note that during each step, each edge $uv$ will be used by at most one additional Steiner tree. 
This is because the edge is added only in the case when either $u$ sends a token to $v$, or the other way around.

We also claim that during all the steps of a given phase, one edge $uv$ is used for a Steiner tree at most $k$ times. This follows from the last bullet point of \cref{obs:bfs}. 
\end{proof}

Finally, we bound the running time. 
We have $O(\log n)$ color classes, each one is constructed in $O(\log n)$ phases, where each phase has $O(\log^2 n)$ steps. 
For each step, we need to run breadth first search for $O(k)$ steps and broadcast information to root via Steiner trees, which takes $O(k \log^3 n \cdot \log n \min(k + \log^2 n))$ where the first term is the diameter of the underlying Steiner tree that is bounded by \Cref{obs:khop_inv_two} and the second term is due to number of Steiner trees per edge that was bounded by \Cref{obs:khop_congestion}. 
This implies running time $O(k \log^8 n \min(k + \log^2 n))$. 
\end{proof}
}

\end{document}